\documentclass[11pt]{article}
\usepackage[a4paper,margin=2.4cm]{geometry}
\usepackage{amsmath,amssymb,amsthm}
\usepackage{graphicx,booktabs,array}
\usepackage[colorlinks=true,linkcolor=blue,citecolor=blue,urlcolor=blue]{hyperref}

\newtheorem{theorem}{Theorem}[section]
\newtheorem{lemma}[theorem]{Lemma}
\newtheorem{proposition}[theorem]{Proposition}
\newtheorem{corollary}[theorem]{Corollary}
\theoremstyle{definition}
\newtheorem{definition}[theorem]{Definition}
\newtheorem{example}[theorem]{Example}
\theoremstyle{remark}
\newtheorem{remark}[theorem]{Remark}
\numberwithin{equation}{section}
\newcommand{\Sm}{[S]}
\newcommand{\Nm}{[N]}
\newcommand{\Qm}{[Q]}

\title{First-order differential constraints\\
for the nonlinear wave equation:\\
wave fronts, admissible nonlinearities and exact solutions}
\author{O.~V. Kaptsov\thanks{Federal Research Center for Information and Computational Technologies, Novosibirsk, Russia. E-mail: profkap@gmail.com. ORCID: 0000-0002-9562-9092.}}
\date{}

\begin{document}
\maketitle

\begin{abstract}
We construct exact solutions of the nonlinear wave equation in three space
dimensions by the method of linear determining equations: first-order
constraint functions are found as solutions of the linear
determining equation, the differential constraints obtained by setting them
equal to zero are integrated, and the resulting ansatz is substituted into
the equation. Lie point symmetries are a special case, and the solutions obtained generically have no symmetry. The
determining equations tie the admissible nonlinearity to the geometry of
the constraint: arbitrary on characteristic hyperplanes, cubic on light
cones, and, for a pair of constraints, exponential on light cones and
quintic on the cylindrical fronts of a line source. Among the solutions are smooth bounded sine-Gordon
walls and vortex configurations whose shape varies arbitrarily along a
front, the fields of point sources moving on arbitrary trajectories for
the cubic and Liouville equations, with an exact nonlinear analogue of
the Larmor formula for the power radiated in the cubic case, and smooth
finite-energy solutions of the cubic equation with any prescribed number of
toroidal nodal surfaces.
\end{abstract}

\section{Introduction}
\subsection{The question and the method}
We study the semilinear wave equation
\begin{equation}\label{eq:main}
  \Box u:=u_{tt}-u_{xx}-u_{yy}-u_{zz}=F(u),\qquad u=u(t,x,y,z),
\end{equation}
the equation of motion of a scalar field with potential $V$, $F=-V'$; it
contains the Klein--Gordon, $\phi^4$, sine-Gordon, Liouville and Tzitz\'eica
models. Exact solutions containing arbitrary functions are rare already in
two independent variables: they occur only for special equations, such as
the Liouville equation, and the theory of Euler, Laplace, Moutard and
Darboux describes when. For the semilinear wave equation in $3+1$
dimensions a general classification of such families is not available.

Our tool is the method of linear determining equations (LDE) \cite{K95,K98}.
One looks for constraint functions
\begin{equation}\label{eq:intro-h}
  h=\xi\cdot\nabla u+q(x)\,u+q_0(x),\qquad
  \xi\cdot\nabla u=\xi^0u_t+\xi^1u_x+\xi^2u_y+\xi^3u_z,
\end{equation}
that solve the LDE: the linearisation of \eqref{eq:main} applied to $h$ must vanish,
not identically, but modulo $h$ and its derivatives (\S\ref{sec:det}). Each solution
$h$ is set equal to zero; integrating the first-order equation $h=0$ (a
differential constraint) gives an ansatz, and substituting the ansatz into
\eqref{eq:main} reduces it to an equation in fewer variables. The procedure is the
same as in the theory of symmetry reductions \cite{Ovs82,Olv93,Ibr85}, where $h$ is
the characteristic of a Lie point symmetry and the linearised equation holds
identically; the LDE has many more solutions. The parameters of
families of solutions of the reduced equation may then be made arbitrary
functions of a characteristic phase (lifting, Lemma~\ref{lem:develop}). All the
solutions of this paper are obtained in this way; symmetry reductions
appear only as the special case of constant data (\S\ref{sec:lie}).

\subsection{Main results}
\begin{enumerate}
\item \emph{The geometry of the wave fronts decides the nonlinearity.} A
  constraint selects a family of wave fronts along which the solution has
  a simple form. If the fronts are planes moving with the characteristic
  speed, every nonlinearity $F$ is allowed (Theorem~\ref{thm:rigidity}). If they are the
  spherical fronts emitted by a moving point, a single homogeneous constraint
  ($q_0=0$ in \eqref{eq:intro-h}) allows only the cubic $F=cu^3$ (Lemma~\ref{lem:weight-null},
  Theorem~\ref{thm:nullk}; for $q_0\neq0$ see Remark~\ref{rem:inhom}), and a pair of constraints also
  the exponential $F=Ce^u$ (Theorem~\ref{thm:liouville}). If they are cylindrical fronts
  emitted by a line, a pair of constraints of which one transports $u$ along
  the rays of the fronts allows only $F=cu^5$ (Theorem~\ref{thm:par}); $5=(n+2)/(n-2)$
  with $n=3$ is the conformal exponent in $2+1$ dimensions, as $3$ is in $3+1$. All
  other single constraints found are Lie point symmetries; whether
  twisting spacelike constraints add anything is open (\S\ref{sec:status}).
\item \emph{Where arbitrary functions come from.} A parameter of a known
  family of solutions may be replaced by an arbitrary function of the
  front variable (for plane fronts, of $t+x$) provided its derivative obeys
  a simple transport equation along the rays (Lemma~\ref{lem:develop}). For plane
  fronts this holds automatically, so any smoothly parametrised family of
  solutions of the transverse two-dimensional equation can be ``lifted'' to
  $3+1$ dimensions with arbitrary functions (Proposition~\ref{prop:lifting}).
\item \emph{Smooth solutions of the sine-Gordon equation.} Lifting the
  two-dimensional sine-Gordon walls, crossing jets and vortices gives
  solutions in $3+1$ dimensions that are smooth and bounded everywhere and
  whose walls turn and shift arbitrarily along the front
  (Proposition~\ref{prop:smoothsg}); a single planar wall with constant parameters is linearly stable
  (Proposition~\ref{prop:stability}). The same construction gives smooth $\phi^4$ walls, Liouville
  lumps and Gaussian packets of arbitrary shape (\S\ref{sec:morelifted}), and lifts the
  Tzitz\'eica solitons, which are singular on curves (\S\ref{sec:lifting}).
\item \emph{Fields of moving sources.} For $\Box u=u^3$ the field of a point
  source moving on an arbitrary trajectory and carrying a dipole axis that
  turns arbitrarily is
  \[
    u=\frac{4\sigma}{3\rho^2-\sigma^2},
  \]
  where $\rho$ is the distance from the source at the moment of emission,
  measured in its rest frame, and $\sigma$ the projection on the dipole axis
  (Remark~\ref{rem:conesol}); for a source at rest it is $4z/(3x^2+3y^2+2z^2)$, and for a
  uniformly accelerated source a rational function of $t,x,y,z$ \eqref{eq:coneaccel}. The field is regular off the trajectory and decays
  like a spherical wave; the power it radiates is given exactly by a
  nonlinear analogue of the Larmor formula, which depends on the
  acceleration relative to the axis and on the rate at which the axis
  turns (Proposition~\ref{prop:radiation}). For the Liouville equation $\Box u=e^u$ the field of a
  source on an arbitrary curve has an arbitrary angular pattern given by a
  meromorphic function (Theorem~\ref{thm:liouville}); it is not a lifted solution,
  and together with the lifted ones it exhausts a natural class of
  Liouville fields, for phases that are real analytic where their rays
  twist (Theorem~\ref{thm:liouvclass}).
\item \emph{Smooth solutions of finite energy.} For every $n\ge0$ the equation $\Box u=u^3$
  has a solution that is smooth on all of space-time, has finite energy,
  decays at infinity and has exactly $n$ toroidal nodal surfaces
  (Proposition~\ref{prop:ladder}).
\item \emph{No symmetry.} For generic data, the families with arbitrary
  functions are invariant under no Lie point symmetry; constant data
  give back the classical symmetry reductions (\S\ref{sec:lie}).
\end{enumerate}

\subsection{Relation to previous work}
Compatible lower-order equations go back to the intermediate integrals of
Monge and Amp\`ere, to Jacobi \cite{Jac62} and Darboux \cite{Gou96,KZ05}; as a method for
exact solutions they were developed by Yanenko and coauthors \cite{Yanenko,SSY}.
The determining equations used here are those of \cite{K95,K98,K02}; in group
analysis such constraints are nonclassical or conditional symmetries
\cite{BC,OR86,OR87,LW,FT,ZTP}, closely related to the direct method \cite{CK,CM,ABH93,Olv94}. For
\eqref{eq:main}, symmetry reductions were classified in \cite{GHW,FT,FSS}, conformal
subalgebras in \cite{BBF98}, conditional symmetries studied in \cite{FSe88,Yeh06},
and the d'Alembert--Hamilton systems behind $u=\varphi(\omega)$ solved in \cite{FZY91,FZR91}.
Arbitrary functions of a null coordinate appear in \cite{GHW,FB96,FS83,Yur97}, always
on families of null hyperplanes; here the phase may also be the retarded
time of an arbitrary curve, for which $\Box\tau\neq0$. For linear waves of this kind
see \cite{KP26,Bat15}. Other constructions of solutions of the
$(3+1)$-dimensional sine-Gordon equation containing arbitrary functions are
given in \cite{ABP09,ABP15}. Exact retarded fields of an arbitrarily moving point
source are classical for the Maxwell equations and known for the
Yang--Mills equations \cite{Kos98,Sar02}; the latter are gauge potentials
built from the velocity of the source and the direction of the ray, and
they give no solutions of a scalar nonlinear wave equation. For a scalar
nonlinear wave equation we are not aware of exact fields of an
arbitrarily moving source.

\subsection{Conventions and terminology}\label{sec:conv}
We work in Minkowski space $M^{1,3}$, $x=(t,x,y,z)$, with
$a\cdot b=a^0b^0-a^1b^1-a^2b^2-a^3b^3$, so that $\Box=\partial_t^2-\Delta_3$ and
$(\nabla f)^2:=\eta^{\mu\nu}\partial_\mu f\,\partial_\nu f$. Indices are raised and lowered with
$\eta=\mathrm{diag}(1,-1,-1,-1)$, repeated indices are summed, and
$A^{(\mu}B^{\nu)}=\tfrac12(A^\mu B^\nu+A^\nu B^\mu)$ denotes symmetrisation. For two vectors
$A,B$ we write $A\odot B$ for the symmetric tensor with components
$(A\odot B)^{\mu\nu}=A^{(\mu}B^{\nu)}$; as a quadratic form, $(A\odot B)(v,v)=(A\cdot v)(B\cdot v)$.
For example, \eqref{eq:o2} below can be written as $S=\lambda\,\eta+\alpha\odot\xi$. With these
conventions $\Box(t-r)=2/r$ and, outside the light cone,
$\Box[(-x\cdot x)^{-1/2}]=[(-x\cdot x)^{-1/2}]^3$, which fixes all signs. The characteristic speed is set to $c=1$: the
equation $c^{-2}u_{tt}-\Delta u=F(u)$ reduces to \eqref{eq:main} by $t\mapsto ct$, so ``light
cone'' and ``speed of light'' stand for the characteristic speed of the
model, e.g.\ the Swihart velocity of a long Josephson junction
(sine-Gordon), the sound speed for dislocations (Frenkel--Kontorova), or
the limiting spin-wave velocity for domain walls in antiferromagnets. All statements
are local and are made away from the singular set of the phase.

A few geometric words are used throughout; we state their meaning for the
wave operator.
\begin{itemize}
\item A \emph{null} or \emph{characteristic} direction has $\xi\cdot\xi=0$; along it the principal
  part of $\Box$ degenerates; $\xi\cdot\xi>0$ is \emph{timelike}, $\xi\cdot\xi<0$ \emph{spacelike}. A
  \emph{characteristic phase} $g$ has $(\nabla g)^2=0$; its level sets are wave fronts, swept
  by straight light rays, the integral curves of $\nabla g$.
\item For a family of light rays the \emph{expansion}, \emph{shear} and \emph{twist} are the
  rates at which a thin bundle grows in cross-sectional area, becomes
  elliptic, and rotates. For $\xi=\nabla g$ null the expansion is $\Box g$ up to a
  factor; the twist vanishes exactly when $\xi$ is proportional to a gradient.
\item A hypersurface is \emph{totally umbilic} if it bends equally in all tangent
  directions (the Hessian of $g$ on its tangent vectors is a multiple of the
  scalar product): planes and spheres in $\mathbb R^3$; hyperplanes, light cones and
  pseudospheres $(x-x_0)\cdot(x-x_0)=\mathrm{const}$ in $M^{1,3}$.
\item The \emph{screen} of a null ray is the set of transverse directions (for a
  light cone, the sphere of directions, $m=2$). A \emph{conformal Killing vector}
  generates a conformal transformation (translation, rotation, boost,
  dilation, inversion); for the cubic these are symmetries.
\end{itemize}

Claims checked by computation or quoted from the literature carry one of three markers:
\Sm\ proved or verified in closed form (symbolic computation in exact
arithmetic), \Nm\ verified numerically, \Qm\ quoted from the literature.
Unmarked statements are proved in place. The details of the numerical
verifications, the remaining longer proofs, among them those of
Theorems~\ref{thm:twist} and~\ref{thm:liouvclass}, and the proofs for hyperboloids and for several
constraints are given in the supplementary material; the proof of
Proposition~\ref{prop:radiation} is in the appendix.

\section{The determining equation}\label{sec:det}
Let $E$ denote \eqref{eq:main} and $[E]$ the system formed by \eqref{eq:main} and all its
differential consequences. Functions of $x$, $u$ and finitely many derivatives
$u_\mu,u_{\mu\nu},\dots$ are differentiated by the total derivatives
\[
  D_\nu=\partial_\nu+u_\nu\,\partial_u+u_{\mu\nu}\,\partial_{u_\mu}+u_{\mu\lambda\nu}\,\partial_{u_{\mu\lambda}}+\cdots,
\]
and $\Box$ applied to such a function is the total d'Alembertian,
$\Box h:=\eta^{\mu\nu}D_\mu D_\nu h$. We consider linear first-order \emph{constraint functions}
\begin{equation}\label{eq:constraint}
  h=\xi^\mu(x)\,u_\mu+q(x)\,u+q_0(x),\qquad u_\mu:=\partial_\mu u .
\end{equation}
\begin{definition}
A constraint function \eqref{eq:constraint} is \emph{admissible} for \eqref{eq:main} if it is a solution of
the linear determining equation (LDE) \cite{K95,K98}
\begin{equation}\label{eq:LDE}
  \Box h-F'(u)\,h=\alpha^\nu D_\nu h+\beta_0(x)\,h\quad\text{on }[E]
\end{equation}
for some $\alpha^\nu(x)$ and $\beta_0(x)$. The equation $h=0$ is then called an
admissible \emph{differential constraint}.
\end{definition}
In detail, for $h$ of the form \eqref{eq:constraint},
\[
  D_\nu h=\partial_\nu\xi^\mu\,u_\mu+\xi^\mu u_{\mu\nu}+\partial_\nu q\,u+q\,u_\nu+\partial_\nu q_0,
\]
\[
  \Box h=2\,\partial^\nu\xi^\mu\,u_{\mu\nu}+\bigl(\Box\xi^\mu+2\,\partial^\mu q\bigr)u_\mu
  +\xi^\mu D_\mu(\Box u)+q\,\Box u+(\Box q)\,u+\Box q_0 .
\]
``On $[E]$'' means that $\Box u$ and its total derivatives are replaced by $F(u)$ and
its total derivatives, in particular $D_\mu(\Box u)=F'(u)\,u_\mu$, and that the
resulting equality holds identically in $x$, $u$, $u_\mu$ and in the second
derivatives $u_{\mu\nu}$ subject only to $\eta^{\mu\nu}u_{\mu\nu}=F(u)$. Since the terms
$\xi^\mu F'(u)u_\mu$ cancel, \eqref{eq:LDE} reads explicitly
\begin{equation}\label{eq:LDEexpl}
\begin{aligned}
  &2\,\partial^\nu\xi^\mu\,u_{\mu\nu}+\bigl(\Box\xi^\mu+2\,\partial^\mu q\bigr)u_\mu
   +q\,F(u)-F'(u)\,(q\,u+q_0)+(\Box q)\,u+\Box q_0\\
  &\qquad=\alpha^\nu\bigl(\partial_\nu\xi^\mu\,u_\mu+\xi^\mu u_{\mu\nu}+\partial_\nu q\,u+q\,u_\nu+\partial_\nu q_0\bigr)
   +\beta_0\,(\xi^\mu u_\mu+q\,u+q_0).
\end{aligned}
\end{equation}
The left side is the linearisation of \eqref{eq:main} applied to $h$; the LDE requires
it to vanish modulo $h$ and its derivatives, not identically. Each admissible $h$
is set equal to zero; integrating the first-order equation $h=0$ gives an
ansatz, which is substituted into \eqref{eq:main} and reduces it to an equation in
fewer variables. Compatibility of the overdetermined system $\{E,h=0\}$ is
not assumed and is not guaranteed a priori; it follows a posteriori whenever
the reduced equation has solutions, as it does for all constraints found
below. For the characteristic of a
Lie point symmetry one may take $\alpha=0$, $\beta_0=0$. Multiplying $h$ by a
nonvanishing function $f(x)$ preserves the form of \eqref{eq:LDE}, with
$\alpha\mapsto\alpha+2\nabla\ln f$; we use this freedom to normalise $\xi$ below.

Write $S^{\mu\nu}:=2\partial^{(\mu}\xi^{\nu)}$, so that the second-order terms on the left of
\eqref{eq:LDEexpl} are $S^{\mu\nu}u_{\mu\nu}$. Sorting \eqref{eq:LDEexpl} by the order of the
derivatives of $u$ gives three groups of conditions; at order two the trace
$\eta^{\mu\nu}u_{\mu\nu}=F(u)$ is not free, which is the origin of $\lambda$ below. With $q_0=0$ for the moment:
\begin{align}
  &\text{order 2:}\quad S^{\mu\nu}=\lambda\,\eta^{\mu\nu}+\alpha^{(\mu}\xi^{\nu)},
  \label{eq:o2}\\
  &\text{order 1:}\quad
    \Box\xi^\mu+2\partial^\mu q=\alpha^\nu\partial_\nu\xi^\mu+q\,\alpha^\mu+\beta_0\,\xi^\mu,
  \label{eq:o1}\\
  &\text{order 0:}\quad (\lambda+q)F=q\,uF'+\gamma\,u,\qquad
    \gamma:=\alpha^\nu\partial_\nu q+\beta_0q-\Box q.
  \label{eq:o0}
\end{align}
Here $\lambda(x)$ is introduced by \eqref{eq:o2} as the coefficient of
the metric; taking the trace of \eqref{eq:o2} gives
\begin{equation}\label{eq:trace}
  2\Box g=4\lambda+(\alpha\cdot\xi)\qquad\text{when }\xi=\nabla g.
\end{equation}
The three groups are coupled only through $\lambda$, $q$, $\gamma$,
and this is what makes the order-by-order analysis effective:
\eqref{eq:o2} is a condition on the geometry of $\xi$ alone,
\eqref{eq:o0} a condition on $F$ alone, and \eqref{eq:o1} the link.
\Sm\ (the expansion \eqref{eq:o2}--\eqref{eq:o0} was rederived
independently.)

\begin{remark}[what is invariant]\label{rem:invariants}
For constraint functions of the form \eqref{eq:constraint}, the three conditions
\eqref{eq:o2}--\eqref{eq:o0} together are equivalent to \eqref{eq:LDE}; each of them
is necessary, and none is sufficient alone. The only freedom is $h\mapsto fh$
with $f(x)\neq0$, under which
\[
  \xi\mapsto f\xi,\quad q\mapsto fq,\quad \lambda\mapsto f\lambda,\quad \gamma\mapsto f\gamma,\quad \alpha\mapsto\alpha+2\nabla\ln f .
\]
Hence the ratios $\lambda/q$ and $\gamma/q$, which alone enter the condition
\eqref{eq:o0} on $F$, and the form $\alpha-2\nabla\ln q$ (where $q\neq0$) are invariant;
all statements below about these quantities are independent of the
normalisation of $\xi$.
\end{remark}
\begin{proposition}[order one in conformal form]\label{prop:star}
For any $\xi$ satisfying \eqref{eq:o2}, with $\alpha$ lowered and
$(i_\xi d\alpha)_\nu=\xi^\mu(\partial_\mu\alpha_\nu-\partial_\nu\alpha_\mu)$, the vector
$\mathcal R_\nu:=\alpha^\mu\partial_\mu\xi_\nu-\Box\xi_\nu$ satisfies
\[
  \mathcal R=d\lambda-\tfrac12\lambda\,\alpha-\tfrac12\,i_\xi d\alpha+(\dots)\,\xi .
\]
Consequently order one \eqref{eq:o1} is equivalent, for
$\psi:=q-\tfrac12\lambda$, to
\begin{equation}\label{eq:star}
  2\,d\psi=\psi\,\alpha-\tfrac12\,i_\xi d\alpha+\hat\beta\,\xi
\end{equation}
for some function $\hat\beta$. \Sm
\end{proposition}
\noindent The proof is a direct computation from the divergence and the trace of
\eqref{eq:o2} (supplementary material, Section~S13).

\subsection{Parameters depending on a characteristic phase}\label{sec:develop}
Most families of solutions with arbitrary functions in this paper come
from one elementary observation. It explains how parameters of a known
family of solutions can be made to depend on a characteristic variable,
and which of them may.

\begin{lemma}[parameters depending on a characteristic phase]\label{lem:develop}
Let $U(x;\tau)$ be a family of solutions of $\Box U=F(U)$ depending
smoothly on a parameter $\tau$, and let $\tau(x)$ be a function with
$(\nabla\tau)^2=0$. Then $u(x)=U\bigl(x;\tau(x)\bigr)$ satisfies
\[
  \Box u-F(u)=2\,\nabla\tau\cdot\nabla h+(\Box\tau)\,h,\qquad
  h:=\partial_\tau U\big|_{\tau=\tau(x)} ,
\]
where $\nabla$ acts on $x$ with $\tau$ fixed. Hence $u$ is a solution if and
only if $h$ obeys the transport law $\nabla\tau\cdot\nabla h+\tfrac12(\Box\tau)h=0$ along the
rays of $\nabla\tau$, i.e.\ the null constraint \eqref{eq:constraint} with
$\xi=\nabla\tau$ and weight $q=\tfrac12\Box\tau$. \Sm
\end{lemma}
\begin{proof}
By the chain rule,
$\Box u=\Box U+2\nabla\tau\cdot\nabla(\partial_\tau U)+(\Box\tau)\partial_\tau U+(\nabla\tau)^2\partial_\tau^2U$.
The last term vanishes because $\tau$ is characteristic, and $\Box U=F(U)=F(u)$.
\end{proof}

The second derivative $\partial_\tau^2U$ drops out because $(\nabla\tau)^2=0$; for a
non-characteristic $\tau$ the term $(\nabla\tau)^2\partial_\tau^2U$ survives and the
substitution generally fails. The lemma does not involve $F$. Its use is to let
the parameters of a known family depend on $\tau$:
\begin{itemize}
\item \emph{Null hyperplanes.} If $\tau$ is a null coordinate, $\Box\tau=0$ and
  $h$ need only be constant along the rays. Any family of solutions that
  do not depend on the other null coordinate qualifies. This gives the
  lifting of \S\ref{sec:lifting}, for every $F$ (Proposition~\ref{prop:lifting}).
\item \emph{Light cones.} For a curve in the symmetry group acting on the field
  $W(\mathbf n)/r$ of a source at rest, the transport law holds for rotations, boosts and
  dilations about the vertex and for time translation along the worldline, but
  not for spatial translations: the source may move on an arbitrary worldline
  with an arbitrarily rotating axis (Remark~\ref{rem:conesol}), and the weight
  $q=\tfrac12\Box\tau=1/\rho$ admits only the cubic (details in the supplementary material).
\end{itemize}
The lemma is not a licence to modulate an arbitrary family: a parameter
may be made to depend on $\tau$ only when $\partial_\tau U$ satisfies the transport law.
Every modulated family in this paper is checked against it, namely
Proposition~\ref{prop:lifting} (where $\Box\tau=0$ and $h$ is constant along the rays),
Remark~\ref{rem:conesol} and Proposition~S1.1 (where $q=1/\rho$), and
the cylinder of \S\ref{sec:par}. The Liouville field of Theorem~\ref{thm:liouville} is
obtained directly from two constraints and does not use the lemma.

\section{Order zero: which nonlinearities}\label{sec:order0}
Equation \eqref{eq:o0} is, at each point $x$, a linear first-order
ordinary differential equation in the variable $u$,
\[
  q(x)\,u\,F'(u)-\bigl(\lambda(x)+q(x)\bigr)F(u)=-\gamma(x)\,u .
\]
Its solution $F$ must be the same function of $u$ at every point, so the
ratios of the coefficients must be constant. This is a condition on the
geometry of $\xi$ as much as on $F$, and it is the main tool of the
following sections (\S\S\ref{sec:null}--\ref{sec:rank2}). For example, on
the light cone $\lambda=2/r$ and $q=1/r+C/r^2$ (Lemma~\ref{lem:weight-null}); the ratio
$\lambda/q=2r/(r+C)$ is constant only for $C=0$. Solving the equation gives
the following list: an arbitrary $F$, the pure powers (massless $\phi^4$ for
$k=3$), powers with a linear term $cu^k+c'u$ (which include $\phi^4$,
$F=u-u^3$), and the logarithmic nonlinearity of the ``gausson''. The
linear term, although allowed by the list, is excluded on every branch
with $\lambda\neq0$ (Corollary~\ref{cor:nolinear}).

\begin{proposition}\label{prop:dichotomy}
Let $\lambda,q,\gamma$ be such that \eqref{eq:o0} admits a solution
$F$ depending on $u$ only. Then, up to the scaling of $u$, exactly one
of the following holds. \Sm
\begin{center}
\begin{tabular}{cccl}
\toprule
$q$ & $\lambda$ & $\gamma$ & $F(u)$\\
\midrule
$=0$ & $=0$ & $=0$ & arbitrary\\
$=0$ & $\neq0$ & any & $F=(\gamma/\lambda)\,u$ (linear)\\
$\neq0$ & $\neq0$ & $=0$ & $c\,u^k$, $k=1+\lambda/q$\\
$\neq0$ & $\neq0$ & $\neq0$ & $c\,u^k+(\gamma/\lambda)\,u$, $k=1+\lambda/q$\\
$\neq0$ & $=0$ & $\neq0$ & $Au+Bu\ln u$, $B=-\gamma/q$\\
$\neq0$ & $=0$ & $=0$ & $F=cu$ (linear)\\
\bottomrule
\end{tabular}
\end{center}
In the third and fourth rows the constancy of $k$ forces $\lambda/q$
to be constant on the domain; in the fourth, $\gamma/\lambda$ must be
constant as well; in the fifth, $\gamma/q$.
\end{proposition}
\begin{proof}
For $q\neq0$ write \eqref{eq:o0} as $uF'-\kappa F=-(\gamma/q)u$
with $\kappa=1+\lambda/q$. The homogeneous solution is $cu^\kappa$.
If $\lambda\neq0$ then $\kappa\neq1$ and a particular solution is
$(\gamma/\lambda)u$; if $\lambda=0$ then $\kappa=1$ and the resonant
particular solution is $-(\gamma/q)u\ln u$. For $q=0$,
\eqref{eq:o0} reads $\lambda F=\gamma u$, which is linear unless
$\lambda=\gamma=0$, in which case it is vacuous.
\end{proof}

The two rows with linear $F$ are recorded for completeness and
discarded from here on.

\begin{theorem}[when $F$ is arbitrary]\label{thm:rigidity}
For $q_0=0$, $F$ is arbitrary if and only if $q=\lambda=0$ (for $q_0\neq0$ it never is,
Remark~\ref{rem:inhom}). In that case one of
the following holds.
\begin{enumerate}
\item[(a)] $\xi$ is null (characteristic). Then $\Box g=0$; the light rays
  along $\xi$ neither spread, nor deform, nor rotate (the congruence is
  expansion, shear and twist free); the level sets of $g$ are null
  hyperplanes $n^\mu(g)x_\mu+c(g)=0$, $n\cdot n=0$, with arbitrary $n(g)$, $c(g)$;
  and the constraint is $\xi^\mu u_\mu=0$: $u$ is constant along the light
  rays (Smirnov--Sobolev \cite{SmSo}).
\item[(b)] $\xi$ is not null. If $\xi$ is proportional to a gradient, the
  level sets are pieces of non-null hyperplanes. In general $\xi$ need not
  be a gradient: this case contains the constraints $K^\mu u_\mu=0$ with $K$
  any generator of a Poincar\'e transformation (a Killing vector:
  translation, rotation, boost, or a combination), admissible for every $F$.
\end{enumerate}
\end{theorem}
In both cases the solutions have a simple meaning. In (b) they are the
Lie-invariant solutions: for $K=\partial_x$ they do not depend on $x$. In (a) the equation reduces, on each null
hyperplane, to the two-dimensional equation in the transverse
variables, and the family of hyperplanes, chosen arbitrarily, carries
arbitrary functions into the solution; this is the source of the lifted
solutions of \S\ref{sec:lifting}.
\noindent The proof (supplementary material, Section~S13) uses the
Sachs equation to exclude twist when $\lambda=0$ and the order-two analysis of
\S\ref{sec:order2} to identify the level sets.

\begin{corollary}[no linear term]\label{cor:nolinear}
On every branch with $\lambda\neq0$ classified below, $\gamma=0$ follows
from the system rather than being imposed; in particular the absence
of a linear term on the shearing branch $k=5$ is a corollary, not a
separate analysis. \Sm
\end{corollary}

\begin{remark}[inhomogeneous constraint]\label{rem:inhom}
With $q_0\neq0$ and $q\neq0$, order zero gives $F=A(u+q_0/q)^k+\dots$: the constraint
only shifts the argument of $F$. With $q=0$ it gives $F=Ce^{(\lambda/q_0)u}$ up to linear
and constant terms for $\lambda\neq0$, and a quadratic polynomial for $\lambda=0$ \Sm. The
exponential row, which contains the Liouville equation, is realised by the
pair of constraints of Theorem~\ref{thm:liouville}. For a single constraint we
set $q_0=0$ from here on.
\end{remark}

The row $q\neq0$, $\lambda=0$ is realised by the multiplicative separation of the
logarithmic nonlinearity, the gausson mechanism \cite{BBM} (supplementary
material, Proposition~S10.2).

\section{Order two: the shape of the level sets}\label{sec:order2}
For $\xi=\nabla g$ the order-two condition \eqref{eq:o2} is a condition on
the matrix of second derivatives $\partial_\mu\partial_\nu g$. Its content is
elementary: restricted to directions tangent to a level set $g=\mathrm{const}$,
this matrix must be a multiple of the scalar product. A surface in
$\mathbb R^3$ with this property bends equally in all directions and is a
plane or a sphere; in $M^{1,3}$ the possibilities are hyperplanes,
hyperboloids $(x-x_0)\cdot(x-x_0)=\mathrm{const}$ and light cones
(Theorem~\ref{thm:umbclass}). The position and size of the level set may
depend on the level $g$; how they may depend on it is decided at orders
one and zero (\S\S\ref{sec:null}--\ref{sec:nonnull}). In differential
geometry such hypersurfaces are called totally umbilic. The quantity
$s=\xi\cdot\xi$ distinguishes characteristic level sets ($s=0$) from timelike
or spacelike normals ($s\neq0$), and $W$ denotes the directions tangent to
the level set and transverse to $\xi$ (the screen, \S\ref{sec:conv}).

Throughout this section $\xi=\nabla g$ and $s:=\xi\cdot\xi$. This is an
assumption on both branches: twisting non-null constraints exist
(Remark~S10.1), and so do twisting null ones
(supplementary material, Proposition~S14.3); the twisting cases are treated in
\S\ref{sec:twist}. By the multiplier freedom of \S\ref{sec:det} any hypersurface-orthogonal
$\xi$ may be normalised to $\nabla g$; the remaining freedom is
$g\mapsto G(g)$. Let $H_{\mu\nu}=\partial_\mu\partial_\nu g$, so that
$S^{\mu\nu}=2H^{\mu\nu}$ and $H^{\mu\nu}\xi_\nu=\tfrac12\partial^\mu s$.

\begin{proposition}\label{prop:order2}
Contracting \eqref{eq:o2} with $\xi$ gives
\begin{equation}\label{eq:ds}
  \partial^\mu s=\Bigl[\lambda+\tfrac12(\alpha\cdot\xi)\Bigr]\xi^\mu+\tfrac12\,s\,\alpha^\mu .
\end{equation}
Consequently: \Sm
\begin{enumerate}
\item[(a)] Null branch, $s\equiv0$: $\lambda=-\tfrac12(\alpha\cdot\xi)$, and
  \eqref{eq:trace} gives $\lambda=\Box g$ for a constraint of rank one.
  At higher rank the trace relation acquires defect terms (supplementary material).
\item[(b)] Non-null branch, $s\neq0$: \eqref{eq:o2} is equivalent to
  the total umbilicity \eqref{eq:umb} of the level sets, with no
  condition on $s$. The coefficients are determined uniquely,
  \[
    \lambda=\tfrac13\bigl(2\Box g-\sigma\bigr),\qquad
    \alpha=\frac2s\,\partial s-\frac{\lambda+\sigma}{s}\,\xi,\qquad
    \sigma:=\xi\cdot\partial\ln|s| .
  \]
  For $s=\varepsilon=\pm1$ they reduce to $\alpha=-(\lambda/\varepsilon)\xi$,
  $2\Box g=3\lambda$.
\end{enumerate}
\end{proposition}
\noindent The proof splits \eqref{eq:o2} into its screen, mixed and $(\xi,\xi)$
components (supplementary material).

\begin{proposition}[umbilicity]\label{prop:umb}
Contracting \eqref{eq:o2} with two screen vectors gives
\begin{equation}\label{eq:umb}
  2H\big|_W=\lambda\,G\big|_W,
\end{equation}
where $G$ is the induced metric: the level sets of $g$ are totally
umbilic. \Sm
\end{proposition}
For example, $g=|\vec x|$ in $\mathbb R^3$ has
$\partial_i\partial_jg=(\delta_{ij}-x_ix_j/|\vec x|^2)/|\vec x|$, which on vectors tangent to the
sphere $|\vec x|=\mathrm{const}$ equals $|\vec x|^{-1}$ times the identity; the
factor $\lambda/2$ in \eqref{eq:umb} plays the role of the curvature
$1/|\vec x|$.

\begin{theorem}\label{thm:umbclass}
A totally umbilic hypersurface of $M^{1,3}$ with constant umbilicity
factor on a connected component is a hyperplane, a pseudosphere
$(x-x_0)\cdot(x-x_0)=\mathrm{const}$, or a light cone. \Qm\
(non-degenerate case: \cite[Ch.~4]{ONeill}; null case: \cite{DB})
\end{theorem}

Constancy of the umbilicity factor on a connected component is
automatic for $\dim W\ge2$ by the Codazzi equation
(supplementary material); for the shearing branch, where the screen is
one-dimensional, it is an extra hypothesis, and this is one reason that
branch needs rank two.

\begin{proposition}[expansion and shear]\label{prop:shear}
On the null branch, $\lambda=0$ if and only if the congruence generated
by $\xi$ is expansion free, and then it is shear free as well. In words:
$\lambda$ measures how fast a thin bundle of light rays spreads, and order
two forbids any anisotropic deformation of the bundle.
\end{proposition}
\begin{proof}
For $\lambda=0$, \eqref{eq:o2} reads $S=\alpha\odot\xi$, whose contraction
with two screen vectors vanishes. The trace of that contraction is the
expansion $\theta$ and its trace-free part the shear $\sigma$; both
vanish simultaneously. The statement ``$\theta\equiv0\Rightarrow
\sigma\equiv0$'' is therefore not an additional lemma but the single
condition $\lambda=0$ read twice.
\end{proof}

\section{One constraint: characteristic hyperplanes and lifting}\label{sec:lifting}
On characteristic surfaces the wave operator loses its derivatives in
the direction of the phase $g$; the reduced equation then involves $g$
only as a parameter, and any parameter of a family of its solutions may
be replaced by an arbitrary function of $g$. Physically, the solution
is a two-dimensional configuration (a wall, a soliton, a vortex
lattice) whose parameters vary along a wave front moving at the speed
of light. The reduced equations are:
\begin{align*}
  &\text{null planes (plane):}&& \Box U(g,y^1,y^2)=-\Delta_\perp U,\\
  &\text{null planes (curved \cite{ZP}):}&& \Box U(g,r,z)=-\bigl(U_{rr}+\tfrac2rU_r+U_{zz}\bigr),\\
  &\text{light cone:}&& \Box\bigl(S(g,\theta^A)/\rho\bigr)=-\rho^{-3}\Delta_{S^2}S,\\
  &\text{shearing cylinder:}&& \Box\bigl(\rho^{-1/2}\Lambda(g,\varphi)\bigr)=-\rho^{-5/2}\bigl(\Lambda_{\varphi\varphi}+\tfrac14\Lambda\bigr).
\end{align*}
In the curved case $r=n'(g)\cdot x+c'(g)$ and $z=b(g)\cdot x$ are built from the
phase data in an adapted Frenet frame of the null curve $n(g)$ \cite{ZP}. Hence:

\begin{proposition}[lifting]\label{prop:lifting}
Let $U(\,\cdot\,;c_1,\dots,c_N)$ be any family of solutions of the reduced
equation of one of these branches (on the cone and cylinder branches
$F=cu^3$ and $F=cu^5$ respectively), depending smoothly on parameters $c$
in an open set $\mathcal U\subset\mathbb R^N$. Then for arbitrary smooth functions
$c_a(g)$ with values in $\mathcal U$ the function $U(\,\cdot\,;c(g))$ solves
$\Box u=F(u)$. The family contributes $N$ arbitrary functions of $g$, one
for each parameter. Additional functional freedom may enter
independently through the choice of the characteristic phase itself
(for example the null curve $n(g)$ and $c(g)$ of a family of null
hyperplanes); the two sources are not to be added in a single count.
\end{proposition}
\begin{proof}
By the displayed identities the phase enters the reduced equation only as
a dummy parameter, so the substitution $c_a\mapsto c_a(g)$ maps solutions to
solutions.
\end{proof}

No integrability of the transverse equation is required, only a
parametrised family of its solutions. The following families lift in
this way (the transverse equation being $\Delta_\perp\psi=\omega(\psi)$, i.e.\ $F=-\omega$, on
the plane branch with $\varsigma=t+x$).
\begin{itemize}
\item \emph{Elliptic sine-Gordon} $\Delta_\perp\psi=\sin\psi$: the Hirota $N$-soliton
  $\psi=4\arctan(G/F)$, $F+iG=\tau=(1+if_1)*\dots*(1+if_N)$,
  $f_j=\exp(y\cos\theta_j+z\sin\theta_j+\eta_j)$, $f_j*f_l=s_{jl}f_jf_l$,
  $s_{jl}=-\tan^2\frac{\theta_j-\theta_l}2$, lifts with $2N$ arbitrary functions $\theta_j(\varsigma),\eta_j(\varsigma)$ \cite{Hir72,K23}; for $N=1$ it is a domain wall whose
  orientation and offset vary arbitrarily along the light front; with real
  data these solutions are smooth (\S\ref{sec:smoothsg}). The doubly periodic families $w=\tan(\psi/4)=G/F$ with elliptic $F$, $G$
\cite{BK88,K23} lift with each of their parameters (translations,
rotation, elliptic moduli) replaced by an arbitrary function.
\item \emph{Generalised separation and Tzitz\'eica solitons}: the separable
  solutions $\psi=a(f(y)+g(z))$ of \cite[Chap.~5]{AKPR} and the elliptic $N$-soliton
  solutions of the Tzitz\'eica equation \cite{K98v} lift in the same way, each
  parameter becoming an arbitrary function of $\varsigma$ (supplementary material,
  Section~S10). For the Tzitz\'eica equation both families are
  singular: the real solutions have logarithmic singularities on curves of the
  transverse plane, so, unlike the sine-Gordon case, no smooth solutions arise.
\item \emph{Liouville} $\Delta_\perp\psi+e^\psi=0$ \cite{Liou53}: $u=\ln\bigl(8|\partial_\zeta h|^2/(1+|h|^2)^2\bigr)$,
  $\zeta=y+iz$, with $h(\zeta,\varsigma)$ subject to $\partial_{\bar\zeta}h=0$ and arbitrary in $\varsigma$. The freedom is
of a different kind from the finitely many functions $c_a(\varsigma)$ above: it is
a function of the two real variables $(\zeta,\varsigma)$, restricted by the
holomorphy condition in $\zeta$, and it is available because the
transverse equation is linearisable. The simplest smooth example, a lump of
arbitrary width and position, is written out in \S\ref{sec:morelifted}.
\item \emph{Gaussons} for $F=Au+Bu\ln u$: the centre of the Gaussian solution lifts
  to an arbitrary trajectory in the transverse plane, a Gaussian packet that
  moves at the speed of light and sways sideways arbitrarily (\S\ref{sec:morelifted}).
\item \emph{Kinks, solitons and profiles}: $\tanh$ and $\sqrt2\,\mathrm{sech}$ walls for
  $F=\pm(u-u^3)$ (\S\ref{sec:morelifted}), the Aubin--Talenti profile \cite{Aubin,Talenti} with $\lambda(g)$, $z_0(g)$, the
  cone mode with $e(g)$ (Remark~\ref{rem:conesol}), the elliptic cylinder
  profiles with $E(g)$ (\S\ref{sec:par}).
\end{itemize}

For lifted sine-Gordon solutions the topological charge $\oint d_\perp u\in8\pi\mathbb Z$
along a closed transverse curve does not depend on $\varsigma$ (supplementary
material, Proposition~S13.2).

\subsection{Smooth solutions of the sine-Gordon equation}\label{sec:smoothsg}
For $F(u)=-\sin u$ the equation $\Box u=F(u)$ is the sine-Gordon equation
$u_{tt}-\Delta u+\sin u=0$, and on the plane branch the transverse equation is the
elliptic sine-Gordon equation $\Delta_\perp\psi=\sin\psi$, which describes, for example,
stationary vortex flows of an ideal fluid with stream function $\psi$
\cite{K23}. Its Hirota solutions with real data are smooth on the whole
plane \cite{K23}, and the lift turns them into smooth solutions in $3+1$
dimensions with arbitrary functions.

\begin{proposition}[smooth lifted jets and vortices]\label{prop:smoothsg}
Let $\theta_j(\varsigma)$, $\eta_j(\varsigma)$, $j=1,\dots,N$, be arbitrary smooth real functions of
$\varsigma=t+x$ such that $\theta_j(\varsigma)-\theta_l(\varsigma)\neq\pi\pmod{2\pi}$ for all $j\neq l$ and \emph{every} $\varsigma$
(so that all $s_{jl}$ are finite), and put
\[
  f_j=\exp\bigl(y\cos\theta_j+z\sin\theta_j+\eta_j\bigr),\qquad
  s_{jl}=-\tan^2\frac{\theta_j-\theta_l}2,
\]
\[
  \mathcal F+i\mathcal G=\sum_{S\subset\{1,\dots,N\}}\ i^{|S|}\prod_{j\in S}f_j\prod_{j<l\in S}s_{jl}.
\]
Then
\begin{equation}\label{eq:smoothsg}
  u=4\arctan\frac{\mathcal G}{\mathcal F}
\end{equation}
solves $u_{tt}-\Delta u+\sin u=0$ and is smooth on all of $M^{1,3}$, with $|u|<2\pi$.
\end{proposition}
\begin{proof}
For fixed $\varsigma$, \eqref{eq:smoothsg} is the Hirota $N$-soliton solution of
$\Delta_\perp\psi=\sin\psi$ \cite{Hir72,K23}, so $u$ solves the equation by
Proposition~\ref{prop:lifting}. $\mathcal F$ is the sum over subsets $S$ of even size
$2m$, with coefficient $i^{2m}\prod s_{jl}$ over $m(2m-1)$ pairs. Since every $s_{jl}\le0$, its
sign is $(-1)^{m+m(2m-1)}=(-1)^{2m^2}=+1$, so every term of $\mathcal F$ is non-negative
and $\mathcal F\ge1$. Hence $\mathcal G/\mathcal F$ is smooth and $u$ is smooth with $|u|<2\pi$. \Sm
\end{proof}

\paragraph{One wall and two walls, explicitly.} Put
\[
  \xi_j=y\cos\theta_j(t+x)+z\sin\theta_j(t+x)+\eta_j(t+x),\qquad j=1,2,
\]
with arbitrary smooth real functions $\theta_j$, $\eta_j$ of the single variable $t+x$.
For $N=1$, \eqref{eq:smoothsg} is
\begin{equation}\label{eq:sg1}
  u=4\arctan\exp\bigl(y\cos\theta(t+x)+z\sin\theta(t+x)+\eta(t+x)\bigr).
\end{equation}
This is a domain wall across which $u$ increases from $0$ to $2\pi$. At a fixed time
the wall is the surface $y\cos\theta(t+x)+z\sin\theta(t+x)+\eta(t+x)=0$: in each plane
$x=\mathrm{const}$ it is a straight line with normal $(\cos\theta,\sin\theta)$ at signed
distance $-\eta$ from the $x$-axis, and this line turns and shifts arbitrarily as $x$
varies. For $\theta=k(t+x)$, $\eta=0$ the wall is a helicoid; for constant $\theta$ and
arbitrary $\eta$ it is a plane carrying a bend of arbitrary shape. In every case the
whole pattern travels along $-x$ at the speed of light, and for constant $\theta,\eta$
it reduces to the static planar kink.

For $N=2$, $\mathcal F=1-s_{12}f_1f_2$ and $\mathcal G=f_1+f_2$, so
\begin{equation}\label{eq:sg2}
  u=4\arctan\frac{e^{\xi_1}+e^{\xi_2}}{1+\tan^2\dfrac{\theta_1(t+x)-\theta_2(t+x)}2\;e^{\xi_1+\xi_2}} .
\end{equation}
The denominator is at least $1$, so $u$ is smooth, with $0<u<2\pi$. Away from the
crossing it reduces to single walls: where $\xi_2\to-\infty$, $u\approx4\arctan e^{\xi_1}$;
where $\xi_2\to+\infty$, $u\approx4\arctan e^{-\xi_1-\delta}$ with $\delta=2\ln\bigl|\tan\frac{\theta_1-\theta_2}2\bigr|$, the
same wall reversed and shifted by $\delta$; and symmetrically in $\xi_1$. In the
fluid picture the level lines in a plane $x=\mathrm{const}$ are two jets crossing at
a saddle point \cite{K23}. In $3+1$ dimensions the two walls cross along the curve
$\xi_1=\xi_2=0$, that is $y=Y_0(t+x)$, $z=Z_0(t+x)$; this curve, like the angle
$\theta_1-\theta_2$ between the walls, can be prescribed arbitrarily, and it travels
along $-x$ at the speed of light. Both formulas were checked against the
$3+1$ equation with sample arbitrary functions (supplementary material). \Sm

For $N\ge3$ the
configuration may contain closed level lines (vortices) between the walls \cite{K23}.
The solutions are bounded but not of finite energy.

\subsection{Further explicit lifted solutions}\label{sec:morelifted}
The solutions below are built in the same way as \eqref{eq:sg1}--\eqref{eq:sg2}. Each plane
$x=\mathrm{const}$ carries a known two-dimensional solution in the variables $(y,z)$,
whose parameters are arbitrary smooth functions of $t+x$; the field depends
on $t$ and $x$ only through $t+x$, so the whole configuration slides along $-x$
at the speed of light (Proposition~\ref{prop:lifting}). At a fixed time a wall becomes a
ruled surface of arbitrary shape, and a localised two-dimensional lump becomes
a tube around the curve $y=c_1(t+x)$, $z=c_2(t+x)$, whose shape and thickness are
arbitrary. All three formulas are smooth on the whole of $M^{1,3}$.

\paragraph{Liouville lumps.} For $\Box u=e^u$ the choice $h=(\zeta-c(\varsigma))/w(\varsigma)$, $\zeta=y+iz$,
$c=c_1+ic_2$, $w>0$, in the lifted Liouville solution of \S\ref{sec:lifting} gives
\begin{equation}\label{eq:liouvlump}
  u=\ln\frac{8\,w(t+x)^2}{\Bigl(w(t+x)^2+\bigl(y-c_1(t+x)\bigr)^2+\bigl(z-c_2(t+x)\bigr)^2\Bigr)^2},
\end{equation}
with three arbitrary real functions $w$, $c_1$, $c_2$ of $t+x$. The solution is smooth,
$u\le\ln(8/w^2)$, and the intensity $e^u$ is a lump of radius $w$ centred at
$(c_1,c_2)$ in each plane $x=\mathrm{const}$, with $\int e^u\,dy\,dz=8\pi$ independently of $t+x$:
the lump moves and changes its size arbitrarily while its total weight is
conserved. A rational $h$ of degree $n$ in $\zeta$ gives $n$ such lumps with
$\int e^u\,dy\,dz=8\pi n$, whose positions and sizes are again arbitrary functions of
$t+x$; then $e^u$ is still smooth, but $u=-\infty$ on the lines where $\partial_\zeta h=0$. For
$\Box u=Ce^u$ replace $u$ by $u-\ln C$. The family follows from
Liouville's general solution \cite{Liou53} by the null lift \cite{GHW,FB96};
\eqref{eq:liouvlump} is its simplest smooth member.

\paragraph{Two lumps.} Up to the factor $2$, $e^u=8|\partial_\zeta h|^2/(1+|h|^2)^2$ is the energy
density of the planar $CP^1$ lump given by the rational map $h$, but here the
lumps have no dynamics of their own: their positions and sizes are
arbitrary functions of $t+x$, there is no force between them, and every
path in the space of parameters is an exact solution of $\Box u=e^u$
(Figure~\ref{fig:lumps}). For $h=\lambda/(\zeta^2-a^2)$, $\lambda\neq0$, the lumps sit at $\pm a$ with
radius $|\lambda/(2a)|$; at $a^2=0$ they merge into a ring of radius $\sqrt{|\lambda|}$, and if
$a^2(t+x)$ changes sign they separate along the perpendicular direction, as in
the right-angle scattering of planar lumps \cite{Ward85}; if $a^2$ only touches zero
they return along the same line. Throughout, $\int e^u\,dy\,dz=16\pi$, and between the
lumps $e^u$ vanishes at the critical point $\zeta=0$ of $h$, where $u=-\infty$. For
$h=2w\zeta/(\zeta^2-a^2)$, whose poles have residues of equal sign, the merger is
different: as $a\to0$ the two lumps fuse into one of radius $2w$, while a
second lump of radius $a^2/(2w)$ shrinks at the origin, with peak value
$e^u\approx32w^2/a^4$. At $a=0$ the degree of $h$ drops to one and a weight $8\pi$
concentrates in a point, so a function $a(t+x)$ that passes through zero
produces a point singularity moving at the speed of light. The outcome of a
merger is thus fixed by the data, not by an interaction of the lumps.

\begin{figure}[t]
\centering\includegraphics[width=\textwidth]{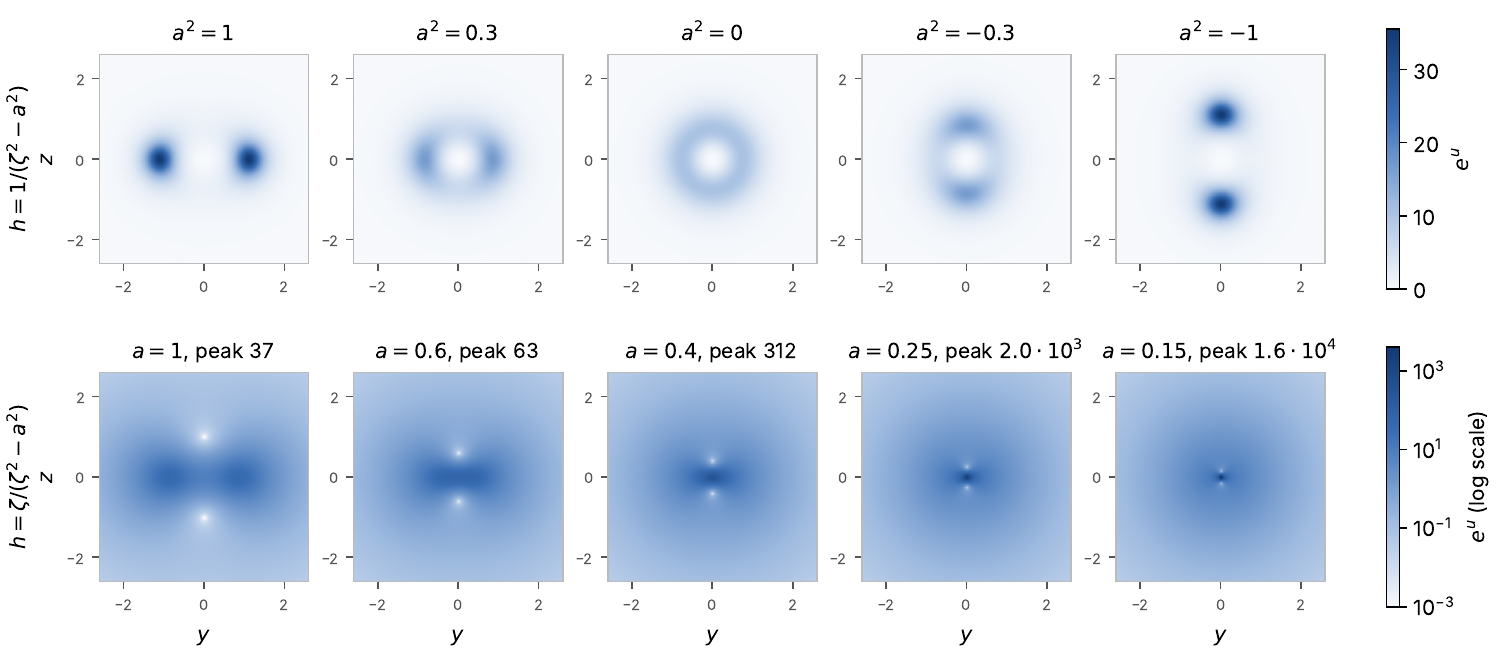}
\caption{Two Liouville lumps: the intensity $e^u$ in a plane $x=\mathrm{const}$ for a
sequence of values of the parameter, i.e.\ of $t+x$. Top: $h=1/(\zeta^2-a^2)$; the lumps
merge into a ring at $a^2=0$ and separate at a right angle. Bottom (logarithmic
scale): $h=\zeta/(\zeta^2-a^2)$; as $a\to0$ a lump of radius $a^2$ collapses at the origin.}
\label{fig:lumps}
\end{figure}

\paragraph{$\phi^4$ walls.} For $\Box u=u-u^3$, the $\phi^4$ model with vacua $u=\pm1$,
\begin{equation}\label{eq:phi4wall}
  u=\tanh\frac{y\cos\theta(t+x)+z\sin\theta(t+x)+\eta(t+x)}{\sqrt2}
\end{equation}
is a domain wall between the two vacua with the same geometry as the
sine-Gordon wall \eqref{eq:sg1}: in each plane $x=\mathrm{const}$ it is a straight line with
normal $(\cos\theta,\sin\theta)$, turning and shifting arbitrarily along $x$. For constant
$\theta,\eta$ it is the static kink, which is linearly stable
(Proposition~\ref{prop:stability}(ii)). For $\Box u=u^3-u$ the same substitution in
$\sqrt2\,\mathrm{sech}$ gives a lifted bright soliton, which is unstable already for constant
parameters (Proposition~\ref{prop:stability}(i)).

\paragraph{Gaussian light bullets.} For the logarithmic equation $\Box u=Au+Bu\ln u$,
$B>0$, of Bia{\l}ynicki-Birula and Mycielski \cite{BBM},
\begin{equation}\label{eq:gausson}
  u=\exp\Bigl(\frac{B-A}B-\frac B4\Bigl[\bigl(y-c_1(t+x)\bigr)^2+\bigl(z-c_2(t+x)\bigr)^2\Bigr]\Bigr)
\end{equation}
is a positive Gaussian packet of fixed amplitude and fixed radius $2/\sqrt B$ (at which
$u$ has dropped by the factor $e$), whose
centre $(c_1,c_2)$ follows an arbitrary path: at a fixed time it is a Gaussian
beam around an arbitrary curve, sliding along $-x$ at the speed of light. At a
fixed $x$ the centre moves in the $(y,z)$-plane with the arbitrary velocity $c'(t+x)$;
this is the speed of a pattern, not of a signal (supplementary material).

\paragraph{Origin of the formulas.} The $3+1$-dimensional solutions of this
subsection and of \S\ref{sec:smoothsg} are obtained from known two-dimensional solutions
by the lift of Proposition~\ref{prop:lifting}. The lift along null hyperplanes as a
mechanism is known \cite{GHW,FB96,Yur97}, and the single walls \eqref{eq:sg1} and
\eqref{eq:phi4wall}, which have the form $u=\varphi(\omega)$ with
$\omega=y\cos\theta(t+x)+z\sin\theta(t+x)+\eta(t+x)$, $(\nabla\omega)^2=-1$, $\Box\omega=0$, are also of the
d'Alembert--Hamilton type \cite{FZY91,FZR91}. What we add is the explicit form of
these solutions in $3+1$ dimensions and their properties: the geometry of the
walls, the smoothness of the multi-wall solutions with independent angles
$\theta_j(t+x)$ (Proposition~\ref{prop:smoothsg}), the conserved weight of the Liouville
lumps, and the linear stability of the unmodulated walls (\S\ref{sec:stability}).

\noindent Formulas \eqref{eq:liouvlump}--\eqref{eq:gausson} were checked by substitution
(supplementary material). \Sm

\section{One constraint: light cones and the cubic}\label{sec:null}
In this section $\xi$ is null, i.e.\ characteristic, and the constraint
prescribes the rate of change of $u$ along light rays. By
Proposition~\ref{prop:order2}(a) and Proposition~\ref{prop:shear} the rays are
straight lines whose bundles spread at the rate $\lambda=\Box g$ without
deforming. If $\lambda=0$ the rays on each level set are parallel and we are in the
Smirnov--Sobolev case of Theorem~\ref{thm:rigidity}. If $\lambda\neq0$ the rays
spread out, and two situations occur: the rays issue from a curve
(light cones of a point source, \S\S\ref{sec:cones}--\ref{sec:nullcurve}), or they twist around
each other (\S\ref{sec:twist}). In both the nonlinearity is forced to be the
cubic (\S\ref{sec:exponent}).

\subsection{Light cones of an arbitrary vertex curve}\label{sec:cones}
\paragraph{Retarded geometry.}
Think of a point source moving along a curve $z(\tau)$ and emitting
signals at the speed of light. A signal that reaches the point $x$ was
emitted at the retarded time $\tau(x)$, from the point $z(\tau)$; the
level sets of $\tau$ are the light cones of the source. Let $z(\tau)$ be a
timelike worldline (a source moving slower than light) parametrised by
proper time,
$V=\dot z$ its four-velocity, $V\cdot V=1$, and $a=\dot V$ its acceleration. For $x$ off the worldline let
$\tau(x)$ be the retarded time, $R=x-z(\tau)$, $R\cdot R=0$, and
\[
  \rho=R\cdot V>0,\qquad \ell=R/\rho,\qquad A=R\cdot a=\rho\,(\ell\cdot a).
\]
Here $\rho$ is the distance from the source in its rest frame at the
moment of emission, $\ell$ the direction of the light ray, and $A$ the
projection of the acceleration onto the ray. Along each ray $R=\rho\ell$
with $\ell$ fixed, so $A$ is linear in $\rho$. The standard identities are
\begin{gather*}
  \partial_\mu\tau=\ell_\mu,\qquad
  \partial_\mu\rho=V_\mu+(A-1)\ell_\mu,\qquad (\nabla\rho)^2=2A-1,\qquad
  \Box\tau=\frac2\rho,\qquad \Box\rho=\frac{4A-2}{\rho},\\
  \partial_\mu \ell_\nu=\frac1\rho\bigl[\eta_{\mu\nu}-\ell_\mu V_\nu-V_\mu \ell_\nu-(A-1)\ell_\mu \ell_\nu\bigr],
\end{gather*}
and in particular $\Box\rho^{-1}=-\rho^{-2}\Box\rho+2\rho^{-3}(\nabla\rho)^2=0$.

\begin{lemma}[transport weight, null branch]\label{lem:weight-null}
Let $g=\tau$, $\xi=\nabla\tau=\ell$. Then: \Sm
\begin{enumerate}
\item \eqref{eq:o2} holds with the unique solution
  $\lambda=2/\rho=\Box g$,
  $\alpha_\mu=-\frac4\rho V_\mu-\frac2\rho(A-1)\ell_\mu$, $\alpha\cdot\xi=-2\lambda$;
\item the general solution of \eqref{eq:o1} is
  \[
    q=\frac1\rho+\frac{C(\tau)}{\rho^2},\qquad
    \beta_0=\frac{2-6A}{\rho^2}+\frac{2\dot C}{\rho^2}-\frac{2C(A-1)}{\rho^3},
  \]
  with $C$ an arbitrary function of the phase;
\item then
  \[
    \gamma=\frac{2C}{\rho^5}\Bigl[\rho\,(1+\dot C)+C\,(1-A)\Bigr],\qquad
    \frac\lambda q=\frac{2\rho}{\rho+C}.
  \]
\end{enumerate}
Each of the conditions $\gamma=0$ and $\lambda/q=\mathrm{const}$ forces
$C\equiv0$. Hence $q=\tfrac12\Box g=1/\rho$ and $k=3$ for every vertex
worldline.
\end{lemma}
\noindent The proof is a direct computation with the identities above
(supplementary material, Section~S13).

In words: the amplitude factor of the constraint must be $q=1/\rho$, so
that $u\sim1/\rho$ decays like a spherical wave. The extra term $C(\tau)/\rho^2$
allowed by order one is a near-field correction, and order zero
forbids it; with $q=1/\rho$ one has $\lambda/q=2$ and $k=3$, whatever the
acceleration of the source.

\begin{remark}[explicit solutions on the cone]\label{rem:conesol}
With $q=1/\rho$ the ansatz is $u=S(\theta^A;\tau)/\rho$ and the reduced
equation is $\Delta_{S^2}S+cS^3=0$ for each $\tau$; the aberration terms
cancel because the acceleration enters as the dipole harmonic
$\kappa=a\cdot\hat n$, annihilated by $\Delta_{S^2}+2$.
Integrating the equation and its product with $S$ over the sphere gives
$c>0$ and $\int_{S^2}S^3=0$: the profile changes sign. The axially
symmetric mode $W(\mu)=4\mu/(3-\mu^2)$ gives, for an arbitrary timelike
worldline and an arbitrary unit spacelike field $e(\tau)\perp\dot z$,
\[
  u=\frac{4\sigma}{3\rho^2-\sigma^2},\qquad \sigma=-R\cdot e,
\]
a solution of $\Box u=u^3$ with five arbitrary functions of one
variable. Verified here in exact arithmetic for an
accelerated, turning worldline with a rotating field $e(\tau)$. \Sm

\emph{Physical meaning.} For a timelike curve, $u=4\sigma/(3\rho^2-\sigma^2)$ is the
field radiated by a point source moving arbitrarily and carrying a dipole
whose axis $e(\tau)$ may turn in an arbitrary way: $u=W(\mu)/\rho$ exactly, with
$W(\mu)=4\mu/(3-\mu^2)$ and $\mu=\sigma/\rho$ the cosine of the angle between the ray and the
axis in the rest frame of the source, an odd, dipole-like angular pattern
decaying like a spherical wave. As for a classical point charge, the field is
singular on the worldline and its total energy is infinite.
\end{remark}

\paragraph{Explicit forms of the cone solutions.} The retarded time is defined
implicitly, but the solutions are explicit in adapted coordinates, and for
simple worldlines also in Minkowski coordinates.
\begin{itemize}
\item \emph{Any worldline.} Let $E_1,E_2,E_3=e$ be orthonormal spacelike vectors along
  the worldline, orthogonal to $V$. Every point off the worldline is
  $x=z(\tau)+\rho\bigl(V(\tau)+m^iE_i(\tau)\bigr)$ with $\rho>0$ and $m\in S^2$ (retarded coordinates). Then
  $\sigma=\rho\,m^3$, and the field of Remark~\ref{rem:conesol} is
  \begin{equation}\label{eq:coneparam}
    u=\frac{4\,m^3}{\rho\,\bigl(3-(m^3)^2\bigr)} ,
  \end{equation}
  for an arbitrary worldline $z(\tau)$ and an arbitrary rotation of the frame $E_i(\tau)$.
\item \emph{Source at rest}, axis along $z$:
  \begin{equation}\label{eq:conestatic}
    u=\frac{4z}{3x^2+3y^2+2z^2},
  \end{equation}
  a static solution of $-\Delta u=u^3$.
\item \emph{Uniformly accelerated source}, $z(\tau)=(\sinh\tau,0,0,\cosh\tau)$, the hyperbola
  $z^2-t^2=1$, $x=y=0$ (other accelerations follow by $u\mapsto\lambda u(\lambda x)$). Here the
  retarded time is found from $t\sinh\tau-z\cosh\tau=P$, and
  $\rho=t\cosh\tau-z\sinh\tau$, with
  \[
    P=\tfrac12\bigl(t^2-x^2-y^2-z^2-1\bigr),\qquad \rho^2=t^2-z^2+P^2 .
  \]
  For the axis along the acceleration, $e=(\sinh\tau,0,0,\cosh\tau)$, one finds
  $\sigma=\tfrac12(x^2+y^2+z^2-t^2-1)$; for a fixed transverse axis $e=\partial_x$, $\sigma=x$. In
  both cases
  \begin{equation}\label{eq:coneaccel}
    u=\frac{4\sigma}{3\,(t^2-z^2+P^2)-\sigma^2}
  \end{equation}
  is a rational solution of $\Box u=u^3$. Since $\rho^2-\sigma^2\ge0$ (for the first axis
  $\rho^2-\sigma^2=x^2+y^2$), the denominator is at least $2\rho^2$, and $u$ is smooth on all of
  $M^{1,3}$ except on the two branches of the hyperbola, the worldline and its
  mirror image, with $|u|\le2/\rho$.
\item \emph{Rotating waves.} For a source at rest and the axis $z$, the profiles
  invariant under rotations about the axis (supplementary material,
  Remark~S13.4) give
  \begin{equation}\label{eq:rotwave}
    u=\frac{f\bigl(\varphi-\alpha(t-r)\bigr)}{\sqrt{x^2+y^2}},\qquad
    f''+f+cf^3=0,\qquad \varphi=\arg(x+iy),
  \end{equation}
  with an arbitrary function $\alpha$ of the retarded time $t-r$: the angular
  pattern turns about the axis in an arbitrary way, and the turning
  propagates outwards at the speed of light. For $c>0$,
  $f=A\,\mathrm{cn}(\omega\theta,k)$ with $\omega^2=1+cA^2$, $k^2=cA^2/(2\omega^2)$; $u$ is single-valued when
  $4K(k)/\omega=2\pi/n$, which fixes $A$ for every $n\ge2$ ($\sqrt c\,A\approx2.0328$ for $n=2$) \Nm. These
  waves are singular on the whole $z$-axis; for nonconstant $\alpha$ they have no
  Lie point symmetry (\S\ref{sec:lie}). For constant $\alpha$, \eqref{eq:rotwave} is a static separated
  solution of the classical form $f(\varphi)/\sqrt{x^2+y^2}$; what is new is the turning by an
  arbitrary function of the retarded time.
\end{itemize}
Formulas \eqref{eq:conestatic}, \eqref{eq:coneaccel} and \eqref{eq:rotwave} were checked
by substitution (supplementary material). \Sm

The acceleration of the source drops out of the field equation, but not out
of the field: it enters the gradient of $u$, and with it the energy carried to
infinity. The next proposition makes this precise. We use the energy-momentum
tensor $T_{\mu\nu}=\partial_\mu u\,\partial_\nu u-\eta_{\mu\nu}\bigl(\tfrac12\,\partial u\cdot\partial u+\tfrac14u^4\bigr)$ of
$\Box u=u^3$; the rest-frame acceleration is the spacelike vector $a\perp V$, and the
rate at which the axis turns, measured against a non-rotating frame carried
by the source (a gyroscope), is the Fermi--Walker derivative $w=\dot e+(e\cdot a)V$,
which is orthogonal to $V$ and to $e$. For spacelike vectors $|a|^2=-a\cdot a$.

\begin{proposition}[radiation of a moving source]\label{prop:radiation}
Let $u=4\sigma/(3\rho^2-\sigma^2)$ be the field of Remark~\ref{rem:conesol} for a timelike
worldline and a unit field $e(\tau)\perp V$, and let $\mu=\sigma/\rho$, $W(\mu)=4\mu/(3-\mu^2)$. \Sm
\begin{enumerate}
\item $u=W(\mu)/\rho$ with $|\mu|\le1$, so $|u|\le2/\rho$; $u$ is smooth off the worldline.
\item Exactly,
  \[
    \partial_\nu u=\frac{f}{\rho}\,\ell_\nu+\frac{1}{\rho^2}\,N_\nu,\quad
    f=-(\ell\cdot a)\,(\mu W)'-(\ell\cdot\dot e)\,W',\quad
    N=-W'\bigl(e+\mu(V-\ell)\bigr)-W\,(V-\ell).
  \]
  The near field $N/\rho^2$ involves only the velocity and the axis; the
  radiation field $f\ell/\rho$ is linear in the acceleration and in $w$.
\item $T_{\mu\nu}=f^2\ell_\mu\ell_\nu/\rho^2+O(\rho^{-3})$, and the energy radiated per unit
  proper time, measured in the rest frame of the source at the moment of
  emission, is
  \[
    P=\oint_{S^2}f^2\,d\Omega
     =C_\parallel\,(e\cdot a)^2+C_\perp\bigl(|a|^2-(e\cdot a)^2\bigr)+C_w\,|w|^2,
  \]
  with $L=\ln(2+\sqrt3)$ and
  \[
    C_\parallel=\tfrac{16\pi}{9}\bigl(5\sqrt3\,L-9\bigr)\approx13.43,\quad
    C_\perp=\tfrac{16\pi}{3}\bigl(3-\sqrt3\,L\bigr)\approx12.05,\quad
    C_w=\tfrac{16\sqrt3\,\pi}{9}\,L\approx12.74 .
  \]
\end{enumerate}
\end{proposition}
\noindent The proof, a direct computation with the identities of \S\ref{sec:cones}
and with the energy-momentum tensor along the rays, is given in Appendix~\ref{app:radiation}; the constants are also checked
symbolically in the supplementary material.

\emph{Examples.} For hyperbolic motion with proper acceleration $g$,
$z(\tau)=g^{-1}(\sinh g\tau,0,0,\cosh g\tau)$, one has $|a|=g$ and $\dot a=g^2V$. With the axis along
the acceleration, $e=a/g$, this gives $(e\cdot a)^2=g^2$ and $w=0$; with a fixed transverse
axis, $e=\partial_x$, it gives $e\cdot a=0$ and $w=0$. Hence
\[
  P=C_\parallel\,g^2\qquad\text{or}\qquad P=C_\perp\,g^2 :
\]
a uniformly accelerated source radiates at a constant rate, as a uniformly
accelerated charge does, and its field is the rational solution \eqref{eq:coneaccel}
(for $g=1$; other $g$ by $u\mapsto\lambda u(\lambda x)$). For uniform circular motion with speed $v$
on a circle of radius $R$ and the axis perpendicular to the plane of the orbit,
$e\cdot a=0$, $w=0$ and $|a|=\gamma^2v^2/R$, $\gamma=(1-v^2)^{-1/2}$, so that
\[
  P=C_\perp\,\frac{\gamma^4v^4}{R^2},
\]
the analogue of synchrotron radiation. In both cases the power depends on the
orientation of the axis through $C_\parallel\neq C_\perp$. \Sm

\emph{Comparison with the Li\'enard--Wiechert field.} For the linear equation
$\Box u=0$ the field of a scalar point charge is $u=q/\rho$, and the same
computation gives $f=q\,(\mathbf n\cdot\mathbf a)$ and $P=\tfrac{4\pi}{3}q^2|a|^2$, the scalar Larmor
formula. Proposition~\ref{prop:radiation} is its exact analogue for the cubic
equation, and it shows four differences. (a) The strength is not a free
parameter: the amplitude of $W$ is fixed by the nonlinearity (the scaling
$u\mapsto ku(kx)$ only maps one worldline to another). (b) There is no monopole:
the angular profile must change sign (Remark~\ref{rem:conesol}), and the simplest field
is dipole-like, with an axis $e$ that is part of the data. (c) The radiation
depends on the direction of the acceleration relative to the axis
($C_\parallel\neq C_\perp$), and a source moving uniformly radiates if its axis turns
($C_w>0$), as a rotating dipole does. (d) Fields of two sources do not
superpose. A source in uniform motion whose axis is carried without
rotation does not radiate; its field is the Lorentz image of the static
solution $4z/(3r^2-z^2)$. As in the linear case, the worldline is prescribed:
the radiated energy is supplied by whatever moves the source, and the back
reaction of the field on the source is not considered here. The energy
density of $\Box u=u^3$ is not positive, because of the term $-\tfrac14u^4$, but this
term decays like $\rho^{-4}$ and does not contribute to $P\ge0$. The formula is
checked in the supplementary material for an accelerated, non-planar
worldline with a turning axis.
\subsection{Sources moving at the speed of light}\label{sec:nullcurve}
The vertex curve may also be null or spacelike. With $\rho=R\cdot\dot z$ the same
construction gives $u=W/\rho$ and $\Box(W/\rho)=-\rho^{-3}\Delta_\Sigma W$, where $\Sigma$ is the set of
light rays issuing from a point of the curve: the round sphere for a timelike
curve, the Euclidean plane for a null curve (a source moving at the speed of
light, whose sphere of directions is flattened by aberration) and the
hyperbolic plane for a spacelike one. The reduced equation is $\Delta_\Sigma W+cW^3=0$
in all three cases, the acceleration again drops out, and only the cubic is
allowed. For a straight null line these are nonlinear analogues of Bateman's
progressing waves \cite{Bat15}. Statements and proofs are in the supplementary
material (Proposition~S1.1, Remarks~S1.2
and~S1.3).

\subsection{The exponent}\label{sec:exponent}
Why the cubic? Along a light ray the spreading rate $\lambda$ obeys a
Riccati-type transport equation. For the light cone of a point,
$\lambda=2/\rho$ and $d\lambda/d\rho=-\tfrac12\lambda^2$: the bundle area grows like $\rho^2$. On the
other hand the amplitude $q$ obeys a linear transport equation (order
one), and order zero requires $q=\mu\lambda$ with a constant $\mu$. A multiple
of a solution of a Riccati equation solves the linear transport
equation only for one value of $\mu$, which is $\mu=\tfrac12$, i.e.\
$k=1+\lambda/q=3$. The following lemma gives the transport equation for
$\lambda$ in general, including rotating (twisting) bundles; in the theory of
relativity it is known as the Raychaudhuri equation.

\begin{lemma}[Raychaudhuri identity]\label{lem:raych}
Let $\xi$ be null and satisfy \eqref{eq:o2}. Then $(\xi\cdot\partial)\xi=\kappa\xi$
for a function $\kappa$, and
\[
  \alpha\cdot\xi=2(\kappa-\lambda),\qquad \partial\cdot\xi=\lambda+\kappa,\qquad
  \xi\cdot\partial\lambda=\kappa\lambda-\tfrac12\lambda^2+\tfrac12\,|d\xi(m,\bar m)|^2 ,
\]
where $m$ is a complex null vector with $m\cdot\xi=0$, $m\cdot\bar m=-1$; the
value $|d\xi(m,\bar m)|$ is independent of the choice of $m$ and vanishes if
and only if $\xi$ is twist free. \Sm
\end{lemma}
\begin{proof}
For $\kappa=0$ this is the Raychaudhuri equation of a shear-free null congruence
in flat space \cite[\S7.1]{PR}. In general, contracting \eqref{eq:o2} with $\xi$ defines $\kappa$ and
gives $\alpha\cdot\xi$; the trace gives $\partial\cdot\xi$; and the divergence of $(\xi\cdot\partial)\xi=\kappa\xi$ gives the
transport law, with $S_{\mu\nu}S^{\mu\nu}=2\lambda^2+2\kappa^2$ from \eqref{eq:o2} and
$\omega_{\mu\nu}\omega^{\mu\nu}=-2\kappa^2+2|d\xi(m,\bar m)|^2$ for $\omega=d\xi$. The identity was also verified
symbolically for the general shear-free congruence in the Kerr representation of
\S\ref{sec:twist} (script \texttt{kerr\_jets\_check.py}).
\end{proof}

The point of the lemma is that the transport law follows from \eqref{eq:o2} alone,
in any normalisation of $\xi$.

\begin{theorem}[the null exponent]\label{thm:nullk}
On the null branch of rank one, let $\lambda\not\equiv0$ and let $F$ be a power
nonlinearity, $F=cu^k$ or $F=cu^k+c'u$, $k\neq1$. Then $k=3$, for every
shear-free geodesic null congruence, twisting or not.
\end{theorem}
\begin{proof}
Order zero forces $q=\mu\lambda$ with $\mu$ constant, so $\psi=q-\tfrac12\lambda=\nu\lambda$
with $\nu=\mu-\tfrac12$. The $\xi$-component of \eqref{eq:star}, with
$\alpha\cdot\xi=2(\kappa-\lambda)$ and $d\alpha(\xi,\xi)=0$, is the transport law
$\xi\cdot\partial\psi=(\kappa-\lambda)\psi$. For $\psi=\nu\lambda$ it reads
$\nu\bigl[\xi\cdot\partial\lambda-\lambda(\kappa-\lambda)\bigr]=0$, and by Lemma~\ref{lem:raych} the bracket
equals $\tfrac12\lambda^2+\tfrac12|d\xi(m,\bar m)|^2$, which is $\ge\tfrac12\lambda^2>0$ at every point where
$\lambda\neq0$, whatever the twist. Hence $\nu=0$, $\mu=\tfrac12$, and $k=1+\lambda/q=3$. Since $k$ is a
constant, one point with $\lambda\neq0$ suffices: no assumption on connectedness or
on the set where $\lambda$ vanishes is needed. The sign of the twist term is
essential, and it is the standard one: twist opposes the focusing of rays
(the optical equations of Sachs \cite[\S7.1]{PR}); it was checked
symbolically on the general Kerr congruence (supplementary material, Lemma~S14.1), with $m$
renormalised to $m\cdot\bar m=-1$.
\end{proof}

\begin{corollary}[the cubic]\label{cor:cubic}
On the null branch with $\lambda\neq0$ and $F=cu^3$, order one is
equivalent to $q=\lambda/2$ together with
\[
  i_\xi d\alpha\propto\xi,\qquad\text{i.e.}\qquad d\alpha(\xi,m)=0 .
\]
\end{corollary}
\noindent Theorem~\ref{thm:twist} below shows that for a twisting congruence this
condition forces a Robinson congruence. As an illustration only,
$d\alpha(\xi,m)$ is of order one for the Kerr congruence and for random quadratic
Kerr functions, while it vanishes for the Robinson congruence and, by
Lemma~\ref{lem:weight-null}, for the light cone (supplementary material). \Nm

For the light cone $d\alpha(\xi,m)=0$ but $d\alpha\neq0$, so the condition
of Corollary~\ref{cor:cubic} is weaker than $d\alpha=0$; the difference is
made by the twist.

\subsection{Twisting null constraints}\label{sec:twist}
If $\xi$ is not proportional to a gradient there are no surfaces
orthogonal to it, and the light rays twist around one another, like the
rays of a vortex. Order two forces the ray family to be geodesic and to
deform without shear, but says nothing about twist; for $\lambda=0$ twist is
excluded by the Sachs equation (proof of Theorem~\ref{thm:rigidity}), and for
$\lambda\neq0$ it is decided at order one by Corollary~\ref{cor:cubic}. All such ray
families are described by Kerr's theorem: one complex function $Y$ solving
$Y_v=YY_\zeta$, $Y_{\bar\zeta}=YY_w$ in the coordinates $w=t-z$, $v=t+z$, $\zeta=x+iy$
\cite[\S7.4]{PR}. A jet computation, given in the supplementary material
(Section~S14), proves:

\begin{theorem}[twisting null constraints]\label{thm:twist}
Let $\xi$ be null, and let $U$ be a connected open set on which $\xi$ twists and $\lambda\neq0$; let $F''\not\equiv0$. On $U$ the constraint $\xi^\mu u_\mu+q u=0$ is admissible for \eqref{eq:main}
if and only if $F=cu^3$, $q=\lambda/2$, and the rays form a Robinson
congruence, $Y=-(A_0+A_2\zeta+A_3w)/(A_1+A_2v+A_3\bar\zeta)$ with
constant $A_i$ (a linear Kerr function). The constraint is then the
conformal symmetry of Proposition~\ref{prop:conformal}, generated by a
null conformal Killing vector.
\end{theorem}
\noindent Theorem~\ref{thm:twist} is local; for a real-analytic congruence, e.g.\ one given by a
holomorphic Kerr function, it extends to the connected component of the
domain by analytic continuation.

\begin{proposition}[conformal symmetries]\label{prop:conformal}
On a rank-one branch (null or not), suppose that the $1$-form $\alpha$ of
\eqref{eq:o2} is closed, $\alpha=-2\nabla\ln f$. Then $K=f\xi$ is a
conformal Killing vector, $2\partial_{(\mu}K_{\nu)}=f\lambda\,\eta_{\mu\nu}$,
and the constraint with $q=\lambda/2$ is
\[
  K^\mu u_\mu+\tfrac14(\partial\cdot K)\,u=0,
\]
the conformal symmetry of $\Box u=c\,u^3$. Hence \eqref{eq:o2}--\eqref{eq:o0}
hold with $F=cu^3$, $\gamma=0$. Closedness of $\alpha$ is independent of the
normalisation of $h$, since $h\mapsto fh$ changes $\alpha$ by $2\nabla\ln f$.
\end{proposition}
\begin{proof}
$2\partial_{(\mu}(f\xi_{\nu)})=f\lambda\eta_{\mu\nu}+(f\alpha+2\nabla f)_{(\mu}\xi_{\nu)}$, and
the second term vanishes for $\alpha=-2\nabla\ln f$; then
$\tfrac14\partial\cdot K=\tfrac12f\lambda$. In four dimensions the scalar of
conformal weight one is invariant under conformal Killing flows exactly
for the cubic.
\end{proof}

Closedness of $\alpha$ is sufficient for admissibility for the cubic but
not necessary: on the light cone $d\alpha\neq0$ and the cubic is carried
nevertheless (Lemma~\ref{lem:weight-null}).

\section{One constraint: hyperboloids, and a summary}\label{sec:nonnull}
For non-null $\xi$ and $\lambda\neq0$ the level sets are hyperboloids $(x-c)\cdot(x-c)=\pm R^2$
(Proposition~\ref{prop:order2}(b)), and only classical symmetries arise: concentric
hyperboloids give the dilation reduction $u=g^{-2/(k-1)}S(\omega)$, hyperboloids with a
moving centre occur only for $k=3$ and give a conformal symmetry, and timelike
twisting constraints are Lie point symmetries by the Herglotz--Noether theorem
\cite{WE67,Hu11}; the spacelike twisting case is open (supplementary material,
Lemma~S8.1, Theorem~S8.2, Proposition~S8.3).

\paragraph{The exponent rule.} On the light-cone branch the exponent obeys
\begin{equation}\label{eq:weight}
  k=1+\frac{\lambda}{q}=1+\frac{4}{m},\qquad m=\frac{2\Box g}{\lambda},
\end{equation}
where $m$ is the number of transverse directions in which the light rays
spread: the amplitude $q$ is half the spreading rate, and the spreading is
shared equally among the $m$ directions. The rule fixes the nonlinearity in
two cases only: for the light cone $m=2$ and $k=3$ (Theorem~\ref{thm:nullk}), and for
a cylinder spreading in one direction, with a second constraint that
transports $u$ along the rays, $m=1$ and $k=5$ (Theorem~\ref{thm:par}). For
hyperboloids $m=3$, but \eqref{eq:weight} is then only a normalisation of $q$:
every $k$ occurs (dilation reduction), and the value $k=7/3$ merely corresponds to
the choice $q=\tfrac12\Box g$. For hyperplanes $\lambda=0$ and $k$ is not fixed.

\paragraph{Summary of rank one.} Table~\ref{tab:branches} collects the
branches. Apart from Poincar\'e symmetries, the dilation, the logarithmic
separation, the null hyperplanes and the light cones of a curve, every
single constraint found is a conformal symmetry of the cubic, detected by
the closedness of a $1$-form (Proposition~\ref{prop:conformal}).

\begin{table}[h]
\centering\small
\begin{tabular}{>{\raggedright}p{2.6cm}>{\raggedright}p{2.6cm}p{1.6cm}>{\raggedright\arraybackslash}p{6.2cm}}
\toprule
constraint & phase & $F$ & ansatz and reduced equation\\
\midrule
$\xi^\mu u_\mu=0$, $q=\lambda=0$ & null planes $n(g)\cdot x+c(g)=0$ & arbitrary &
 $u=U(g;y^1,y^2)$ constant along the rays; $\Delta_\perp U=-F(U)$; lifting (\S\ref{sec:lifting})\\
$q=1/\rho$ & light cone of a curve & $c\,u^3$ &
 $u=W/\rho$, $\Delta_\Sigma W+cW^3=0$ on $\Sigma=S^2$, $\mathbb R^2$ or $\mathbb H^2$; e.g.\ $u=4\sigma/(3\rho^2\mp\sigma^2)$\\
Robinson field & twisting rays & $c\,u^3$ & conformal symmetry; $-\Delta_HV+V=cV^3$ on $S^3$ (\S\ref{sec:cr})\\
$s=\pm1$ & concentric hyperboloids & $c\,u^k$, any $k$ &
 $u=g^{-2/(k-1)}S(\omega)$, dilation reduction\\
moving centre & pencil (S8.1) & $c\,u^3$ & conformal symmetry\\
$h=u_z+q(z)u$ or null & planes & $Au+Bu\ln u$ & multiplicative separation (Proposition~S10.2)\\
\bottomrule
\end{tabular}
\caption{One constraint. All reduced equations verified by substitution. \Sm}
\label{tab:branches}
\end{table}

\section{Two constraints}\label{sec:rank2}
\subsection{Normal form and the reduction conditions}\label{sec:rank2normal}
Let $h_a=\xi_a\cdot\nabla u+q_au+q_{0a}$, $a=1,\dots,r$, be constraint functions with linearly
independent fields $\xi_a$. The logic is the same as for one constraint: we
look for solutions $h_1,\dots,h_r$ of the system of linear determining equations \cite{K02}
\begin{equation}\label{eq:QDEsys}
  \Box h_a-F'(u)\,h_a=\sum_b\bigl(\alpha_{ab}^\nu D_\nu h_b+\beta_{ab}(x)\,h_b\bigr)\quad\text{on }[E],
\end{equation}
set all $h_a$ equal to zero, integrate the resulting first-order system and
substitute the ansatz into \eqref{eq:main}; compatibility of $\{E,h_1=\dots=h_r=0\}$ is
again not assumed but follows from the existence of solutions of the reduced
equation. As for a single constraint, the
principal part gives $S_a=\lambda_a\eta+\sum_b\alpha_{ab}\odot\xi_b$, so each field may
now be corrected by the others. We assume that the distribution spanned by
the $\xi_a$ is integrable, as it is in all configurations below; let
$\omega_1,\dots,\omega_{4-r}$ be independent functions with $\xi_a\cdot\nabla\omega_i=0$. Integrating the
constraints by characteristics gives
\begin{equation}\label{eq:ansatz-r}
  u=A(x)\,\varphi(\omega_1,\dots,\omega_{4-r})+B(x),
\end{equation}
with $A$, $B$ particular solutions of the transport equations, and
\[
  \Box u=A\,(\nabla\omega_i\cdot\nabla\omega_j)\,\varphi_{ij}+\bigl(A\,\Box\omega_i+2\nabla A\cdot\nabla\omega_i\bigr)\varphi_i
  +(\Box A)\,\varphi+\Box B .
\]
The ansatz reduces \eqref{eq:main} to an equation for $\varphi(\omega)$ if and only if,
after dividing by the appropriate power of $A$, all coefficients are
functions of $\omega$. For $F=cu^k$, $k\neq0,1,2$, this means $B=0$ and
\begin{align}
  &\text{(R1)}\quad \nabla\omega_i\cdot\nabla\omega_j=A^{k-1}\,\hat g_{ij}(\omega),\label{eq:R1}\\
  &\text{(R2)}\quad \Box\omega_i+2\nabla\ln A\cdot\nabla\omega_i=A^{k-1}\,b_i(\omega),\label{eq:R2}\\
  &\text{(R3)}\quad \Box A=A^k\,\gamma(\omega),\label{eq:R3}
\end{align}
with reduced equation $\hat g_{ij}\varphi_{ij}+b_i\varphi_i+\gamma\varphi=c\varphi^k$. These are the conditions of
orders two, one and zero of \eqref{eq:QDEsys}, written for the ansatz. For
$r=3$ there is one variable $\omega$ and \eqref{eq:R1}--\eqref{eq:R3} reduce $\Box u=cu^k$
to an ordinary differential equation (\S\ref{sec:rank3}).

\begin{proposition}[which nonlinearities admit a weight]\label{prop:weights}
Let \eqref{eq:ansatz-r} reduce $\Box u=F(u)$, and suppose that $A$ or $B$ is not a
function of $\omega$. Then, up to translations and scalings of $u$, $F$ is one
of the following, with the stated normalisation of the weights:
\begin{center}
\small
\begin{tabular}{lll}
\toprule
$F(u)$ & weights & order-zero condition\\
\midrule
$c\,u^k+c'u$ & $B=0$, $A^{k-1}=N$ & $\Box A=c'A+A^k\gamma(\omega)$\\
$C\,e^u+c''$ & $A=1$, $e^B=N$ & $\Box B=c''+e^B\gamma(\omega)$\\
$A_0u+B_0u\ln u$ & $B=0$, $N=1$ & $\Box A/A-B_0\ln A=\gamma(\omega)$\\
$c\,u^2$ & $N=A$, $A$ and $B$ free & two conditions\\
\bottomrule
\end{tabular}
\end{center}
Here $N$ is the common factor in \eqref{eq:R1}. For every other $F$, in particular
$F=\pm\sin u$ and the Tzitz\'eica nonlinearity $e^u-e^{-2u}$, one may take $A=1$,
$B=0$, and the reduction conditions are the d'Alembert--Hamilton system
$\nabla\omega_i\cdot\nabla\omega_j=\hat g_{ij}(\omega)$, $\Box\omega_i=b_i(\omega)$ \cite{FZY91,FZR91,Yeh06}. \Sm
\end{proposition}
\noindent The proof reduces the reduction condition to a functional equation for
$F''$ along the level sets of $\omega$ (supplementary material,
Section~S13).

For a single constraint this is the content of Proposition~\ref{prop:dichotomy};
for several constraints it says that only one of them can carry a weight
when $F$ is a power, the exponential or $u\ln u$. For $\sin u$ and for Tzitz\'eica
no constraint can carry a weight; all new solutions that we found for
these equations come from the lifting mechanism of \S\ref{sec:lifting}.

\subsection{Three types}\label{sec:types}
For $r=2$ the metric $\hat g_{ij}$ of \eqref{eq:R1} on the plane of the variables
$(\omega_1,\omega_2)$ has one of three types.
\begin{itemize}
\item \emph{Hyperbolic}: two real characteristic phases, $(\nabla\sigma_1)^2=(\nabla\sigma_2)^2=0$.
  The reduced equation is a wave equation in $1+1$ dimensions.
\item \emph{Elliptic}: one complex characteristic phase $Y$, $(\nabla Y)^2=0$,
  $M:=\nabla Y\cdot\nabla\bar Y\neq0$. The reduced equation is elliptic in $(Y,\bar Y)$.
\item \emph{Parabolic}: one characteristic phase $\sigma$ and a second variable $\omega$
  with $\nabla\sigma\cdot\nabla\omega=0$. The reduced equation contains no second derivative in $\sigma$.
\end{itemize}

\subsection{The parabolic type}\label{sec:par}
\begin{theorem}[the characteristic phase is a parameter]\label{thm:par}
In the parabolic type, for $F=cu^k$ ($k\neq0,1$), $F=Ce^u$, or trivial weights
($A=1$, $B=0$), $b_\sigma=0$ in \eqref{eq:R2} and $aN^m$ is constant along every ray of $\sigma$,
$m=2/(k-1)$ ($m=0$ in the last two cases). Hence the reduced equation is an ordinary
differential equation in $\omega$ in which $\sigma$ enters only as a parameter, and
any family of its solutions may be modulated arbitrarily in $\sigma$. In particular,
if on an open set of rays exactly one principal curvature of the level sets of
$\sigma$ vanishes ($\kappa_1\neq0=\kappa_2$, e.g.\ $\sigma=t-\sqrt{x^2+y^2}$), then $\nabla\omega$ has no
component along the flat direction, $F=Ce^u$ and $A=1$ are excluded, and $F=cu^5$.
\end{theorem}
\begin{proof}
Along a ray of $\sigma$, with affine parameter $s$, the functions $\sigma$ and $\omega$ are
constant. The cross-sectional area of a thin bundle of rays changes by the
factor $a(s)=(1-\kappa_1s)(1-\kappa_2s)$, where $\kappa_{1,2}$ are the principal curvatures of
the level set of $\sigma$ in the screen directions, and
$N=-(\nabla\omega)^2=\sum_iw_i^2/(1-\kappa_is)^2$ with constants $w_i$. Since $\Box\sigma=d\ln a/ds$ and,
by \eqref{eq:R1}, $2\nabla\ln A\cdot\nabla\sigma=m\,d\ln N/ds$ with $m=2/(k-1)$ ($m=0$ for the
exponential and for arbitrary $F$, where $A=1$), condition \eqref{eq:R2} for $i=\sigma$
reads
\[
  \frac{d}{ds}\ln\bigl(a\,N^m\bigr)=b_\sigma\,N .
\]
Both sides are rational in $s$, and the left side has only simple poles. If
$\kappa_iw_i\neq0$ for some $i$, $N$ has a double pole at $s=1/\kappa_i$, so $b_\sigma=0$. Otherwise $N$
is constant along the ray, $d\ln a/ds=-\sum_i\kappa_i/(1-\kappa_is)$ must be constant, hence
$\kappa_1=\kappa_2=0$, both sides vanish and again $b_\sigma=0$. For $\kappa_1\neq0=\kappa_2$ put
$x=1-\kappa_1s$; then $aN^m=x^{1-2m}(w_1^2+w_2^2x^2)^m$, and its logarithmic derivative in
$x$, $\bigl[(1-2m)w_1^2+w_2^2x^2\bigr]/\bigl(x(w_1^2+w_2^2x^2)\bigr)$, vanishes identically only if
$w_2=0$ and $m=\tfrac12$, i.e.\ $k=5$.
\end{proof}

The theorem also decides the geometry: $aN^m$ must be constant along every
ray, and this leaves three possibilities for the principal curvatures.
\begin{center}
\small
\begin{tabular}{llll}
\toprule
level sets of $\sigma$ & curvatures & $F$ & solutions\\
\midrule
null hyperplanes & $\kappa_1=\kappa_2=0$ & arbitrary & lifted two-dimensional families (\S\ref{sec:lifting})\\
light cones of a curve & $\kappa_1=\kappa_2$ & $c\,u^3$ & $u=W/\rho$ with $\Delta_\Sigma W+cW^3=0$ (\S\ref{sec:null})\\
one curvature zero (cylinders) & $\kappa_1\neq0=\kappa_2$ & $c\,u^5$ & e.g.\ $u=\rho^{-1/2}\Lambda(\varphi;\sigma)$, $\Lambda''+\tfrac14\Lambda+c\Lambda^5=0$\\
\bottomrule
\end{tabular}
\end{center}
For the cone, $a\propto(1-\kappa s)^2$ and $N\propto(1-\kappa s)^{-2}$ give $m=1$, $k=3$; for the
cylinder, $a\propto1-\kappa s$ and $N\propto(1-\kappa s)^{-2}$ give $m=\tfrac12$, $k=5$. This is the rule
\eqref{eq:weight} derived from the determining equations. The cylinder needs two
constraints: the phase $t-\rho$, $\rho=\sqrt{x^2+y^2}$, is the wave front of a line
source along the $z$-axis, a bundle of its rays spreads in the angular
direction but not along $z$, and no single constraint is admissible;
adding $u_z=0$ removes the $z$-direction. The same conclusion follows from a
defect formula for the trace of the order-two condition
proved in the supplementary material. For the logarithmic
nonlinearity $N=1$, so $\kappa_iw_i=0$ along every ray and \eqref{eq:R2} only fixes $aA^2$;
for the cylinder $\sigma=t-\rho$ condition \eqref{eq:R3} then forces $B_0=0$ (script
\texttt{log\_cylinder\_check.py}).

\paragraph{The cylinder field explicitly.} Let $\rho=\sqrt{x^2+y^2}$, $\varphi=\arg(x+iy)$ and
$\sigma=t-\rho$. The two constraints give $u=\rho^{-1/2}\Lambda(\varphi;\sigma)$, independent of $z$, with
\begin{equation}\label{eq:cylode}
  \Lambda''+\tfrac14\Lambda+c\Lambda^5=0,\qquad
  \Lambda'^2+\tfrac14\Lambda^2+\tfrac c3\Lambda^6=E .
\end{equation}
The substitution $w=\Lambda^{-2}$ turns the first integral into $w'^2=4Ew^3-w^2-\tfrac43c$, whence
\begin{equation}\label{eq:cylwp}
  \Lambda^2=\frac{E}{\wp(\varphi-\varphi_0)+\frac1{12}},\qquad
  g_2=\frac1{12},\qquad g_3=\frac1{216}+\frac{4c}{3}E^2,
\end{equation}
where $\wp$ is the Weierstrass function with invariants $g_2,g_3$; by
Theorem~\ref{thm:par} the constants $E$ and $\varphi_0$ may be arbitrary functions of $\sigma$. For
$c=0$, $\wp(\varphi)=-\tfrac1{12}+\tfrac14\sin^{-2}(\varphi/2)$ and $\Lambda=2\sqrt E\sin(\varphi/2)$, the linear
cylindrical wave. The zeros of $\Lambda$ are the poles of $\wp$. Qualitatively, $\Lambda$
moves in the potential $\tfrac14\Lambda^2+\tfrac c3\Lambda^6$, and the period of an oscillation of
amplitude $\Lambda_m$ is
\[
  T=4\int_0^{\pi/2}\frac{d\vartheta}{\sqrt{\frac14+\frac c3\Lambda_m^4\,(1+\sin^2\vartheta+\sin^4\vartheta)}} .
\]
\begin{itemize}
\item $c>0$. Every solution oscillates, and $T$ decreases strictly from $4\pi$ (small
  amplitude) to $0$ (large amplitude). The field is single-valued around the
  axis only if $T=2\pi/n$; for each $n=1,2,\dots$ this fixes one energy $E_n$, and
  $\Lambda_n$ has $2n$ simple zeros on the circle. For $c=1$, $E_1\approx0.7834$, $E_2\approx6.394$,
  $E_3\approx21.60$ \Nm; for other $c>0$ replace $\Lambda$ by $c^{-1/4}\Lambda$ and $E_n$ by $c^{-1/2}E_n$. The solutions
  \begin{equation}\label{eq:cylfield}
    u=\frac{\Lambda_n\bigl(\varphi-\varphi_0(t-\rho)\bigr)}{(x^2+y^2)^{1/4}}
  \end{equation}
  are cylindrical waves of a line source along the $z$-axis: they decay like
  $\rho^{-1/2}$, are singular only on the axis, and their nodal set consists of the
  $2n$ surfaces $\varphi=\varphi_0(t-\rho)+\mathrm{const}$, spiral sheets that wind around the axis
  as prescribed by the arbitrary function $\varphi_0$ and expand at the speed of light.
  For nonconstant $\varphi_0$ their only symmetry is the translation along the axis
  (\S\ref{sec:lie}).
\item $c<0$. The potential has maxima at $\Lambda=\pm\Lambda_*$, $\Lambda_*^4=1/(4|c|)$. The bounded
  oscillations, $|\Lambda|<\Lambda_*$ and $E<1/(12\sqrt{|c|})$, have $T>4\pi$ and are never
  single-valued; all other nonconstant solutions blow up at a finite angle. The only single-valued
  profile is the constant $\Lambda_*$, which gives the static solution
  $u=\pm(4|c|)^{-1/4}(x^2+y^2)^{-1/4}$.
\end{itemize}
If single-valuedness is not required, for example in a wedge $\varphi_1<\varphi<\varphi_2$,
the energy $E(\sigma)$ may also vary arbitrarily. For constant $E$ and $\varphi_0$, \eqref{eq:cylfield}
is a static separated solution of the classical form $\rho^{-1/2}\Lambda(\varphi)$, and its
integration \eqref{eq:cylwp} is elementary; what is new is the dependence on the
retarded time $t-\rho$ of the line source, which is not a null coordinate. Formula \eqref{eq:cylwp} and the
reduction \eqref{eq:cylfield} with an arbitrary $\varphi_0$ were checked symbolically
(supplementary material). \Sm

\subsection{The elliptic type for powers}\label{sec:ellpow}
For $F=cu^k$, $k\neq0,1,2$, and a complex phase $Y$ that is the direction function
of a shear-free family of rays (a Kerr function), the elliptic type occurs only
when the rays meet a fixed straight line, $u=W(\mathbf n)\,r^{-2/(k-1)}$, or for a lift;
twisting ray families carry no such reduction. Complex phases that are not
Kerr direction functions do reduce the equation: the field $u=4\sigma/(3\rho^2-\sigma^2)$
of an accelerated source is an example. The precise statement, its proof and
these examples are in the supplementary material (Theorem~S11.3,
Remark~S11.4); a classification for general complex phases is
open.

\subsection{The exponential: the Liouville field of a moving source}\label{sec:liouville}
For $F=Ce^u$ the weight sits in $B$ (Proposition~\ref{prop:weights}), and the
condition $B=B(Y,\bar Y)$ of \S\ref{sec:ellpow} disappears. The rays may then issue
from an arbitrary curve. Let $z(s)$ be a smooth curve, $s(x)$ the retarded
parameter defined by $R\cdot R=0$, $R=x-z(s)$, $R^0>0$, and put
\[
  R_0=t-z^0(s),\qquad \mathbf n=\frac{\mathbf x-\mathbf z(s)}{R_0},\qquad Y=\frac{n_1+in_2}{1+n_3},
\]
so that $R_0$ is the distance, in the laboratory frame, from the retarded
position of the source and $Y$ is the stereographic coordinate of the
direction of the ray.

\begin{theorem}[Liouville field of an arbitrary source]\label{thm:liouville}
For every smooth curve $z(s)$, on an open set where the retarded parameter $s(x)$
is smooth and $R\cdot\dot z\neq0$:
\begin{enumerate}
\item $(\nabla Y)^2=0$, $\Box Y=0$, $\nabla Y\cdot\nabla\bar Y=-(1+|Y|^2)^2/(2R_0^2)$ and
  $\Box\ln R_0=-1/R_0^2$;
\item the two constraints
  \begin{equation}\label{eq:liouville-h}
    h_1=\nabla s\cdot\nabla u+\frac{2}{\dot z\cdot R}=0,\qquad h_2=\dot z^\mu(s)\,\partial_\mu u=0
  \end{equation}
  have the general solution $u=w(Y,\bar Y)-2\ln R_0$;
\item this $u$ solves $\Box u=Ce^u$ if and only if
  $(1+|Y|^2)^2\,w_{Y\bar Y}+Ce^w=2$, the Liouville equation on the sphere
  $\Delta_{S^2}w+Ce^w=2$.
\end{enumerate}
Consequently, for every meromorphic function $f$ and $C>0$,
\begin{equation}\label{eq:liouville}
  u=\ln\frac2C+\ln\frac{|f'(Y)|^2\,(1+|Y|^2)^2}{\bigl(1+|f(Y)|^2\bigr)^2}-2\ln R_0
\end{equation}
solves $\Box u=Ce^u$; for $C<0$ replace $\ln(2/C)$ by $\ln(2/|C|)$ and $1+|f|^2$ by $1-|f|^2$,
on the set where $|f|\neq1$.
\Sm\ (the solution \eqref{eq:liouville} was also checked to $25$ digits on an
accelerated timelike curve with an arbitrary $f$.)
\end{theorem}
\begin{proof}
(1) With $\partial_\mu s=R_\mu/(R\cdot\dot z)$ one finds
$\partial_\mu n_i=R_0^{-1}\bigl[\delta^i_\mu-n_i\delta^0_\mu-(\dot z^i-n_i\dot z^0)\partial_\mu s\bigr]$, and since
$(\nabla s)^2=0$ and $\nabla s\cdot(\nabla x^i-n_i\nabla t)=0$, $\nabla n_i\cdot\nabla n_j=(n_in_j-\delta_{ij})/R_0^2$,
independently of $\dot z$. Hence $(\nabla Y)^2=-R_0^{-2}(\nabla_{S^2}Y)^2=0$, $Y$ being a conformal coordinate
on the sphere, and $\nabla Y\cdot\nabla\bar Y=-(1+|Y|^2)^2/(2R_0^2)$. Locally $Y$ is a holomorphic or
antiholomorphic function of the Kerr function of the shear-free family of rays issuing
from the curve; a solution of the Kerr system (S14.1) has $\Box Y=4(Y_{wv}-Y_{\zeta\bar\zeta})=0$, and with
$(\nabla Y)^2=0$ this gives $\Box Y=0$. Finally $\Box s=2/(R\cdot\dot z)$, $\Box R_0=-\dot z^0\Box s$ and
$(\nabla R_0)^2=1-2\dot z^0R_0/(R\cdot\dot z)$ give $\Box\ln R_0=-R_0^{-2}$. These computations use only
$R\cdot\dot z\neq0$, so (1) holds for timelike, spacelike and null curves.
The constraint function $h_2$ differentiates along the velocity and $h_1$ along the ray;
both annihilate $Y$ and $\ln R_0$ except $\nabla s\cdot\nabla\ln R_0=1/(\dot z\cdot R)$, which gives (2). With
(1), $\Box(w-2\ln R_0)=2w_{Y\bar Y}\nabla Y\cdot\nabla\bar Y+2/R_0^2=\bigl(2-(1+|Y|^2)^2w_{Y\bar Y}\bigr)/R_0^2$, while
$Ce^u=Ce^w/R_0^2$, which gives (3). The general solution of the Liouville
equation on the sphere is classical \cite{Liou53}.
\end{proof}

The second constraint fails the order-two condition on its own for an
accelerated curve, and the first is an inhomogeneous constraint of the
exponential row of Remark~\ref{rem:inhom}; only together do they reduce the
equation. For a source at rest, \eqref{eq:liouville} is the singular solution
$\ln(2/(Cr^2))$ of the Gelfand equation $-\Delta u=Ce^u$ dressed with an angular
factor; for an accelerated source it has no Lie symmetry (\S\ref{sec:lie}). The
intensity $e^u\propto e^{w(\mathbf n)}/R_0^2$ obeys an inverse-square law with an
arbitrary angular pattern, and its total angular weight is quantised,
$\int_{S^2}e^w\,d\Omega=8\pi\deg f/C$ for rational $f$. In the language of
\S\ref{sec:ellpow}, these are the twist-free rays from the points of a real curve.
We do not classify the elliptic type for the exponential; whether, apart
from the lift, these are its only reductions is part of the open problem
stated in \S\ref{sec:status}.

\paragraph{Explicit forms of the Liouville field.} In the retarded coordinates
$x=z(s)+R_0\,(1,\mathbf n)$, $R_0>0$, $\mathbf n\in S^2$ (the laboratory direction of the ray),
$Y=(n_1+in_2)/(1+n_3)$, so the right side of \eqref{eq:liouville} is an explicit function
of $(s,R_0,\mathbf n)$ for every worldline and every meromorphic $f$. For a source at rest
and $f(Y)=Y^n$, $n=1,2,\dots$, one has $|Y|^2=(r-z)/(r+z)$ and
\begin{equation}\label{eq:liouvstatic}
  e^u=\frac{8n^2\,(x^2+y^2)^{n-1}}{C\,\bigl[(r+z)^n+(r-z)^n\bigr]^2},
\end{equation}
a static solution of $-\Delta u=Ce^u$: for $n=1$ the Gelfand solution $e^u=2/(Cr^2)$,
for $n=2$ $e^u=8(x^2+y^2)/\bigl(C(r^2+z^2)^2\bigr)$. For $n\ge2$ the intensity vanishes on the
$z$-axis, and its angular weight is $8\pi n/C$. \Sm

\paragraph{The Liouville field is not a lifted solution.} The lift of
\S\ref{sec:lifting} also solves $\Box u=e^u$, with $u=\ln\bigl(8|\partial_\zeta h|^2/(1+|h|^2)^2\bigr)$ and
$h(\zeta,\varsigma)$ holomorphic in $\zeta$ and arbitrary in $\varsigma$ \cite{Yur97,FS83}. The two families
meet only for a source moving at the speed of light along a straight line;
for every other curve the field of Theorem~\ref{thm:liouville} has no translation
symmetry, whereas every lifted solution has one (supplementary material,
Section~S12).

\paragraph{A common construction.} Both families are instances of the
following elementary statement. Its proof is a two-line computation that
does not depend on the signature of the metric (with $\Box$ replaced by $\Delta$ it
gives solutions of $\Delta u=-e^u$ from harmonic morphisms into the sphere), and
we do not claim it as new; what is specific to Minkowski space, and new,
is the classification of the phases for which it applies,
Theorem~\ref{thm:liouvclass} below.

\begin{proposition}[Liouville fields from complex null phases]\label{prop:liouvgen}
Let $\Phi$ be a complex function with
\[
  (\nabla\Phi)^2=0,\qquad \Box\Phi=0,\qquad N:=\nabla\Phi\cdot\nabla\bar\Phi<0,\qquad \Box\ln(-N)=0 .
\]
Then $u=\ln\bigl(-4N/(1+|\Phi|^2)^2\bigr)$ solves $\Box u=e^u$. The conditions are preserved by
$\Phi\mapsto g(\Phi)$ for every meromorphic $g$.
\end{proposition}
\begin{proof}
With $P=1+|\Phi|^2$ the first two conditions give $\Box P=2N$ and $(\nabla P)^2=2|\Phi|^2N$, hence
$\Box\ln P=2N/P^2$ and $\Box u=\Box\ln(-N)-2\Box\ln P=-4N/P^2=e^u$. For $g(\Phi)$ one has
$(\nabla g)^2=g'^2(\nabla\Phi)^2$, $\Box g=g''(\nabla\Phi)^2+g'\Box\Phi$ and $\ln(-N)\mapsto\ln(-N)+\ln g'+\ln\bar g'$,
and $\ln g'(\Phi)$ is again annihilated by $\Box$.
\end{proof}

The lifted family is $\Phi=h(\zeta,\varsigma)$, with $N=-2|h_\zeta|^2$. The family of
Theorem~\ref{thm:liouville} is $\Phi=f(Y)$ with $Y$ the direction of the ray from the
source; there $\ln(-N)=\ln|f'|^2+2\ln(1+|Y|^2)-2\ln R_0+\mathrm{const}$, and
$\Box\ln(1+|Y|^2)=-1/R_0^2=\Box\ln R_0$ by Theorem~\ref{thm:liouville}(1). 

\paragraph{There are no other Liouville fields of this form.} The first two
conditions of Proposition~\ref{prop:liouvgen} say that $\Phi$ is a complex
\emph{functionally invariant} phase in the sense of Smirnov and Sobolev \cite{SmSo}
and Bateman \cite{Bat15}: every $F(\Phi)$ solves the linear wave equation. Such
phases are the complex-valued harmonic morphisms of Minkowski space, and by
Baird and Wood (real-analytic case) \cite{BW98,BW03} and Wehbe ($C^2$ case)
\cite{Weh10} they are carried by shear-free families of straight light rays. We
use this in the following form.

\begin{lemma}[functionally invariant phases are carried by rays {\cite{BW98,Weh10}}]\label{lem:fip}
Let $\Phi\in C^3$ be complex, with $(\nabla\Phi)^2=0$, $\Box\Phi=0$ and $\nabla\Phi\cdot\nabla\bar\Phi\neq0$ on an open set.
Then the two real null directions $l,n$ with $l\cdot\nabla\Phi=n\cdot\nabla\Phi=0$ generate
shear-free families of rays, and on an open dense subset at least one of the
two families consists of straight lines. Hence there $\Phi$ is constant along
the rays of a shear-free family of \emph{straight} light rays.
\end{lemma}
\noindent The proof, a computation with the commutators of the Newman--Penrose
derivatives, is given in the supplementary material
(Section~S12).

The lemma and the formulas below were also checked numerically
(supplementary material). \Nm

\begin{theorem}[classification of the Liouville fields of Proposition~\ref{prop:liouvgen}]\label{thm:liouvclass}
Let $\Phi$ satisfy the four conditions of Proposition~\ref{prop:liouvgen} on an open set,
and assume $\Phi$ real analytic where the rays of Lemma~\ref{lem:fip} twist. Then on
an open dense subset, after a Poincar\'e transformation, either
\begin{enumerate}
\item[(a)] $\Phi=h(\zeta,\varsigma)$ with $\varsigma$ a null coordinate and $h$ holomorphic in $\zeta$: the
  lifted Liouville field; or
\item[(b)] $\Phi=f(Y)$ with $Y$ the direction of the rays from a real curve and $f$
  meromorphic: the field of a moving source, Theorem~\ref{thm:liouville}.
\end{enumerate}
Twisting ray families give nothing beyond (a).
\end{theorem}
\noindent The proof is given in the supplementary material (Section~S12). For twist-free rays order two leaves null
hyperplanes and the light cones of a curve, and the conditions reduce to
holomorphy; for twisting rays the condition $\Box\ln(-N)=0$ forces $\Psi_L$ to be
constant by a CR argument on the three-dimensional space of rays, which gives
case (a).

\section{Three constraints: reduction to an ordinary differential equation}\label{sec:rank3}
For $r=3$ the ansatz \eqref{eq:ansatz-r} depends on one variable, $u=A(x)\varphi(\omega)$,
and \eqref{eq:R1}--\eqref{eq:R3} reduce the equation to an ordinary differential
equation. The three constraints are the derivatives tangent to the level
sets of $\omega$, with the weight $-\nabla\ln A$. For an arbitrary $F$ the weight
must be trivial (Proposition~\ref{prop:weights}), and the conditions are the
d'Alembert--Hamilton system $(\nabla\omega)^2=a(\omega)$, $\Box\omega=b(\omega)$, whose general
solution was found by Fushchych, Zhdanov and Revenko \cite{FZR91}; for the
cubic the weight enlarges the class considerably.

\subsection{Weighted d'Alembert--Hamilton systems}\label{sec:wdh}
The ansatz $u=\varphi(\omega)$ of \cite{FZY91,FZR91} corresponds, in the language of
this paper, to three homogeneous constraints tangent to the level sets
of $\omega$. The constraints \eqref{eq:constraint} carry a weight $q$, and the
corresponding ansatz is $u=A(x)\,\varphi(\omega)$.

\begin{proposition}[weighted reduction]\label{prop:wdh}
The ansatz $u=A(x)\varphi(\omega)$ reduces $\Box u=u^3$ to an ordinary differential
equation $a\varphi''+b\varphi'+c\varphi=\varphi^3$ if and only if
\begin{equation}\label{eq:wdh}
  \frac{(\nabla\omega)^2}{A^2}=a(\omega),\qquad
  \frac{\Box\omega+2\nabla\ln A\cdot\nabla\omega}{A^2}=b(\omega),\qquad
  \frac{\Box A}{A^3}=c(\omega).
\end{equation}
For constant $A$ this is the d'Alembert--Hamilton system. If $a\not\equiv0$,
the substitution $A\mapsto Af(\omega)$, $\varphi\mapsto\varphi/f$, $\omega\mapsto h(\omega)$ with
$f=h'|a|^{1/2}$ and $h''/h'=-(b+a')/(3a)$ brings \eqref{eq:wdh} to the
canonical form
\begin{equation}\label{eq:wdhcan}
  (\nabla\omega)^2=\varepsilon A^2,\qquad \nabla\cdot(A^2\nabla\omega)=0,\qquad \Box A=c(\omega)A^3,
  \qquad \varepsilon=\pm1,
\end{equation}
with reduced equation $\varepsilon\psi''+c(\omega)\psi=\psi^3$.
\end{proposition}

\begin{theorem}[the conformal $4$-Laplacian]\label{thm:plap}
In canonical form the weight is fixed, $A=|\nabla\omega|:=(\varepsilon(\nabla\omega)^2)^{1/2}$, and
\eqref{eq:wdhcan} is equivalent to
\begin{equation}\label{eq:plap}
  \nabla\cdot\bigl((\nabla\omega)^2\nabla\omega\bigr)=0,\qquad \Box|\nabla\omega|=c(\omega)\,|\nabla\omega|^3 .
\end{equation}
Thus the ans\"atze with $(\nabla\omega)^2\not\equiv0$ reducing $\Box u=u^3$ to a second-order ordinary differential equation
are the Lorentzian $4$-harmonic functions $\omega$ (solutions of the $4$-Laplace
equation, the Euler--Lagrange equation of $\int(\nabla\omega)^4\,d^4x$) whose gradient
length solves the conformal cubic up to a function of $\omega$, and the
solutions are $u=|\nabla\omega|\,\psi(\omega)$. A geometric reformulation in the conformally flat metric $A^2\eta$ is given in the
supplementary material. \Sm
\end{theorem}

\subsection{Smooth global solutions}\label{sec:cr}
\paragraph{Conformal compactification.} The change of variables
$\tau\pm\chi=2\arctan(t\pm r)$, $u=\Omega v$ with
$\Omega=2\bigl[(1+(t+r)^2)(1+(t-r)^2)\bigr]^{-1/2}$, maps all of Minkowski space onto the
bounded region $|\tau|+\chi<\pi$ of the cylinder $\mathbb R\times S^3$, with $\tau$ a new time
and $\chi$ the polar angle on the unit three-sphere. Because the cubic is
conformally invariant, $\Box u=u^3$ becomes
\[
  v_{\tau\tau}-\Delta_{S^3}v+v=v^3 .
\]
Infinity of $M^{1,3}$ becomes a set of ordinary points of the cylinder, so a $v$
that is smooth on the cylinder gives a $u$ that is smooth on all of
space-time and decays like $\Omega\sim r^{-2}$. Write the points of $S^3$ as
$(z_1,z_2)\in\mathbb C^2$, $|z_1|^2+|z_2|^2=1$, with $z_1=n_0+in_3$, $z_2=n_1+in_2$ and
$n=\tfrac\Omega2(1+t^2-r^2,2\vec x)$. The circles $z\mapsto e^{i\phi}z$ (Hopf circles) fill $S^3$;
$\xi_H$ is their unit tangent field, and $\Delta_H=\Delta_{S^3}-\xi_H^2$ is the Laplacian with
the second derivative along these circles removed. The twisting null constraint of Remark~S13.3 is the
invariance under the flow $\partial_\tau+\xi_H$.

\paragraph{Finite energy; the simplest solutions.} If $v$ is smooth on the closed region
$|\tau|+\chi\le\pi$, then $u=\Omega v=O(r^{-2})$ and $\nabla u=O(r^{-3})$ at fixed $t$, so the energy
$E=\int\bigl(\tfrac12u_t^2+\tfrac12|\nabla u|^2-\tfrac14u^4\bigr)\,d^3x$ is finite and conserved; this is
the familiar use of conformal compactification for conformally invariant
equations \cite{Chr86}. The simplest examples are homogeneous on $S^3$, $v=\psi(\tau)$
with $\psi''+\psi=\psi^3$ (supplementary material, Example~S13.5), a classical reduction
invariant under $\mathfrak{so}(4)$. At $t=0$ one has $\tau=0$, $\partial_t\tau=\Omega$ and
$\int\Omega^4\,d^3x=\int|\nabla\Omega|^2\,d^3x=2\pi^2$, whence
\begin{equation}\label{eq:Epsi}
  E=2\pi^2\Bigl(\tfrac12\psi'^2+\tfrac12\psi^2-\tfrac14\psi^4\Bigr),
\end{equation}
$2\pi^2$ times the energy of the oscillator. Explicitly,
\begin{equation}\label{eq:omegasn}
  u=\pm\Omega=\frac{\pm2}{\sqrt{(1+(t+r)^2)(1+(t-r)^2)}},\qquad
  u=\Omega\,A\,\mathrm{sn}\Bigl(\sqrt{1-\tfrac{A^2}2}\;\tau,\ k\Bigr),\quad k^2=\frac{A^2}{2-A^2},
\end{equation}
with $0<|A|<1$ and energies $\pi^2/2$ and $\pi^2A^2(1-A^2/2)$; as $A\to1$ the second tends to
$u=\Omega\tanh(\tau/\sqrt2)$, again with $E=\pi^2/2$. The solution $u=\Omega$ is positive; the
others change sign. \Sm

\begin{theorem}[Robinson reduction and the CR Yamabe equation]
\label{thm:cryamabe}
For every smooth function $V$ on $S^3$,
\[
  \Box\Bigl[\Omega\,V\bigl(e^{-i\tau}z\bigr)\Bigr]
  =\Omega^3\,\bigl[(-\Delta_H+1)V\bigr]\bigl(e^{-i\tau}z\bigr),
  \qquad \Delta_H:=\Delta_{S^3}-\xi_H^2 ,
\]
where $\Delta_H$ is the sub-Laplacian of the standard CR structure. Hence the
solutions of $\Box u=cu^3$ invariant under this null conformal Killing
vector are exactly $u=\Omega\,V(e^{-i\tau}z)$ with
\begin{equation}\label{eq:cry}
  -\Delta_HV+V=c\,V^3\quad\text{on }S^3,
\end{equation}
the CR Yamabe equation of the standard CR sphere (homogeneous dimension
$Q=4$, critical exponent $(Q+2)/(Q-2)=3$). \Sm
\end{theorem}
\noindent The proof uses the conformal covariance $\Box_\eta(\Omega v)=\Omega^3(\Box_{g_E}v+v)$;
the second derivatives along the Hopf circles cancel (supplementary material).

For $c>0$ the CR Yamabe equation has infinitely many sign-changing
solutions \cite{MM12} \Qm; each gives a global smooth solution of $\Box u=u^3$ on
$M^{1,3}$, decaying like $\Omega$. An explicit one follows.

\begin{example}[an explicit sign-changing solution]\label{ex:torus}
For $V$ invariant under the torus $U(1)\times U(1)$ acting on $(z_1,z_2)$,
$V=W(\mu)$ with $\mu=|z_1|^2-|z_2|^2$, equation \eqref{eq:cry} with $c=1$ becomes
\[
  4\bigl((1-\mu^2)W'\bigr)'=W-W^3,\qquad -1\le\mu\le1 .
\]
In Minkowski coordinates $\mu=1-8(x^2+y^2)/P$ with
$P=(1+(t+r)^2)(1+(t-r)^2)$, and
\[
  u=\Omega\,W\Bigl(1-\frac{8(x^2+y^2)}{P}\Bigr)
\]
solves $\Box u=u^3$ for every solution $W$ regular at $\mu=\pm1$ (the reduction
was checked symbolically). The equation has no rational or algebraic
solutions of the simplest types (apart from the constants $0$ and $\pm1$, the only one,
$W=\operatorname{sign}\mu$, is singular), but it has an odd regular solution with
\[
  W(1)=-W(-1)=4.2994680439,\qquad W(0)=0,\qquad W'(0)=2.6945497,
\]
found by shooting \Nm; the identity
$-4\int(1-\mu^2)W'^2=\int(W^2-W^4)=-72.560$ is satisfied to seven digits.
A practical representation of $W$ is the
expansion in Legendre polynomials, the eigenfunctions of the linear part:
\[
  W(\mu)=3.52794\,P_1(\mu)+0.65981\,P_3(\mu)+0.09668\,P_5(\mu)+0.01309\,P_7(\mu)+\dots,
\]
whose truncation after $P_7$ is accurate to $2\cdot10^{-3}$, and after $P_{15}$ to $5\cdot10^{-7}$ \Nm.
The resulting $u$ is a global smooth solution of $\Box u=u^3$ on
$M^{1,3}$ with nodal set $8(x^2+y^2)=P$ (supplementary material, Figure~S1). It is
Hopf invariant, hence also invariant under the conformal time
$\partial_\tau$; solutions that genuinely rotate along the fibres require
$\xi_HV\neq0$ and lead to a partial differential equation on $S^3/U(1)$.
\end{example}

These solutions are a classical reduction, by three commuting conformal
generators (supplementary material); what is new is the global analysis below.

\begin{lemma}[the profiles are global]\label{lem:profile}
Put $\mu=\tanh s$, $s\in\mathbb R$, so that the profile equation $4((1-\mu^2)W')'=W-W^3$ reads
$W_{ss}=q(s)(W-W^3)$ with $q=1/(4\cosh^2s)$.
\begin{enumerate}
\item Every solution exists for all $s\in\mathbb R$, i.e.\ on the whole interval $-1<\mu<1$, and
  $|W_s|\le\sqrt{2F(0)}$, where $F=\tfrac12W_s^2+\tfrac14q(W^2-1)^2$.
\item As $s\to\pm\infty$, $W_s\to\beta_\pm$ and $W=\alpha_\pm+\beta_\pm s+o(1)$; equivalently
  $W=\alpha_\pm\pm\tfrac12\beta_\pm\ln\frac2{1\mp\mu}+o(1)$ as $\mu\to\pm1$.
\item For every $A\in\mathbb R$ there is exactly one solution bounded at $\mu=1$ with $W(1)=A$;
  it is real analytic in $\mu$ near $\mu=1$ and depends analytically on $A$. A solution
  is bounded at $\mu=1$ if and only if $\beta_+=0$; otherwise
  $W\sim\tfrac12\beta_+\ln\frac2{1-\mu}$. The same holds at $\mu=-1$ with $\beta_-$. Hence $u=\Omega W(\mu)$ is
  smooth on $M^{1,3}$ if $\beta_+=\beta_-=0$; if $\beta_+\neq0$ it has a logarithmic singularity on
  the $z$-axis ($\mu=1$), and if $\beta_-\neq0$ on the expanding ring $z=0$, $x^2+y^2=1+t^2$
  ($\mu=-1$).
\end{enumerate}
\end{lemma}
\begin{proof}
(1) Along solutions $F_s=\tfrac14q_s(W^2-1)^2=-2\tanh s\,(F-\tfrac12W_s^2)$, so $|F_s|\le2F$ and $F$ is
finite on every bounded interval; hence $W_s$ and $W$ cannot blow up at finite $s$, and
every solution is global. Since $q_s$ has the sign of $-s$, $F\le F(0)$ on $\mathbb R$.
(2) By (1), $|W|=O(|s|)$, and $q\le e^{-2|s|}$, so $W_{ss}=O(|s|^3e^{-2|s|})$ is integrable; hence $W_s\to\beta_\pm$
and $W-\beta_\pm s\to\alpha_\pm$. The second form follows from $s=\pm\tfrac12\ln\frac2{1\mp\mu}+o(1)$.
(3) Near $\mu=1$ put $x=e^{-2s}=(1-\mu)/(1+\mu)$; the equation becomes
$4(xW_x)_x=(1+x)^{-2}(W-W^3)$. If $\beta_+=0$, then $W\to\alpha_+$ and $xW_x=-\tfrac12W_s\to0$, so
\[
  W(x)=\alpha_++\tfrac14\int_0^x\frac{d\xi}{\xi}\int_0^\xi\frac{W-W^3}{(1+\eta)^2}\,d\eta .
\]
Since $\bigl|\int_0^x\xi^{-1}\int_0^\xi h\bigr|\le|x|\sup|h|$, the right side is a contraction on $\{|W-\alpha_+|\le1\}$
for $|x|\le\delta(|\alpha_+|)$, real or complex. Hence for each $A$ there is exactly one solution
bounded at $\mu=1$ with $W(1)=A$; it is analytic in $x$, hence in $\mu$, and depends
analytically on $A$, uniformly on compact sets. If $\beta_+\neq0$, $W$ is unbounded by (2).
The point $\mu=-1$ follows by $\mu\mapsto-\mu$. If $W$ is analytic on $[-1,1]$, then $V=W(\mu)$ is
smooth on $S^3$, because $\mu=|z_1|^2-|z_2|^2$ is a polynomial there, and $u=\Omega V$ is smooth
on $M^{1,3}$. Finally $\mu=1$ is the $z$-axis, and, since $P=(1+t^2-r^2)^2+4r^2$, $\mu=-1$
means $(1+t^2-r^2)^2+4z^2=0$, i.e.\ $z=0$, $x^2+y^2=1+t^2$.
\end{proof}

\begin{proposition}[a ladder of smooth solutions]\label{prop:ladder}
Let $c=1$ and consider $4((1-\mu^2)W')'=W-W^3$ on $[-1,1]$ with $W$ bounded at $\mu=\pm1$.
\begin{enumerate}
\item It has infinitely many solutions $\pm W_k$ with $E(W_k)\to\infty$, where
  $E(W)=\int_{-1}^1\bigl(2(1-\mu^2)W'^2+\tfrac12W^2-\tfrac14W^4\bigr)d\mu$; each gives a solution
  $u=\Omega W(\mu)$ of $\Box u=u^3$ smooth on $M^{1,3}$, whose energy is $\pi^2E(W)$.
\item Every positive solution is $W\equiv1$ (every negative one $W\equiv-1$); all others change sign.
\item For every $n\ge0$ it has a solution with exactly $n$ zeros in $(-1,1)$. The
  corresponding solution of $\Box u=u^3$ is smooth on $M^{1,3}$, has finite energy, and
  its zero set at each time consists of $n$ disjoint nested embedded tori,
  surfaces of revolution about the $z$-axis (round at $t=0$).
\end{enumerate}
\end{proposition}
\begin{proof}
The equation is the Euler--Lagrange equation of $E$ on the torus-invariant part of
$H^1(S^3)$, on which the measure of $S^3$ is uniform in $\mu$. The embedding
$H^1(S^3)\subset L^4(S^3)$ is compact, the cube being subcritical in three
dimensions, so the symmetric mountain-pass theorem \cite{AR73} gives (1);
critical points on the invariant subspace are critical points on
$H^1(S^3)$ by Palais' principle of symmetric criticality, and smooth by elliptic
regularity. Part (2) is the rigidity theorem of Bidaut-V\'eron and V\'eron
\cite{BVV91} for $-\Delta u+\lambda u=u^q$ on $S^n$ with $\lambda\le n/(q-1)$, here $\lambda=1\le\tfrac32$.

(3) We shoot from $\mu=1$. For $A>0$ let $W(\cdot\,;A)$ be the solution of
Lemma~\ref{lem:profile}(3) with $W(1)=A$. Then $(W,W_s)(0;A)$ depends continuously on $A$, so $F(0;A)$
and with it the bounds $|W|\le C(1+|s|)$, $|W_{ss}|\le C(1+|s|)^3e^{-2|s|}$ are uniform for $A$ in
compact sets; by dominated convergence $\alpha(A)=\alpha_-$ and $\beta(A)=\beta_-$ are continuous, and
the remainders of Lemma~\ref{lem:profile}(2) are uniform. $W(\cdot\,;A)$ is regular at both ends
exactly when $\beta(A)=0$. Its zeros are simple, and their number $N(A)$ in $(-1,1)$ is
finite, since $W\to A$ at $\mu=1$ and $W-\alpha-\beta s\to0$ at $\mu=-1$ with $(\alpha,\beta)\neq(0,0)$: if $\beta=0$,
then $\alpha\neq0$, because the solution bounded at $\mu=-1$ with value $0$ is $W\equiv0$.

\emph{Local behaviour of $N$.} Fix $A_*>0$. For $A$ near $A_*$,
$|W_s-\beta(A)|+|W-\alpha(A)-\beta(A)s|\le C(1+|s|)^3e^{2s}$ for $s\le0$, with $C$ independent of $A$, and
$W(s;A)\ge A/2$ for $s\ge S_+$. If $\beta(A_*)\neq0$, then $|W|\ge\tfrac12|\beta(A_*)||s|-C'$ for $s\le-S$, and $N$
is constant near $A_*$. If $\beta(A_*)=0$, then $\alpha_*=\alpha(A_*)\neq0$; choose $S$ with $W(\cdot\,;A_*)\neq0$ on
$(-\infty,-S]$ and $2C(1+S)^4e^{-2S}<\tfrac12|\alpha_*|$. Two zeros $s_1<s_2\le-S$ of $W(\cdot\,;A)$ would give,
by Rolle, $|\beta(A)|\le C(1+|s_2|)^3e^{2s_2}$, and then $|\alpha(A)|\le2C(1+|s_2|)^4e^{2s_2}<\tfrac12|\alpha_*|$, which is
impossible for $A$ near $A_*$. Hence near $A_*$
\[
  N(A)=N(A_*)+\varepsilon(A),\qquad \varepsilon(A)=1\ \text{if}\ \operatorname{sign}\beta(A)=\operatorname{sign}\alpha_*,\ \varepsilon(A)=0\ \text{otherwise},
\]
and in particular $N(A_*)\le N(A)\le N(A_*)+1$.

\emph{$N(1)=0$ and $N(A)\to\infty$.} $W(\cdot\,;1)\equiv1$. For large $A$ put $W=Aw$ and $\sigma=s-\ln A$;
then $w''=qw-A^2qw^3$ with $A^2q=(e^{\sigma}+e^{-\sigma}/A^2)^{-2}\le e^{-2\sigma}$, $A^2q\to e^{-2\sigma}$, and $q\le A^{-2}e^{-2\sigma}$.
The solution $w(\cdot\,;A)$ is the fixed point of $w=1+\int_\sigma^\infty(\tau-\sigma)(qw-A^2qw^3)\,d\tau$, a
contraction on $\{|w-1|\le1\}$ over $[\sigma_0,\infty)$ with $\sigma_0$ independent of $A\ge1$; so $w\to w_\infty$ in
$C^1[\sigma_0,\infty)$, and by continuous dependence in $C^1[-L,\sigma_0]$ for every $L$, where $w_\infty$ solves
$w''+e^{-2\sigma}w^3=0$, $w\to1$ as $\sigma\to+\infty$. With $r=e^{-\sigma}$ this is the radial Lane--Emden
equation $w_{rr}+w_r/r+w^3=0$ in the plane, $w(0)=1$, $w_r(0)=0$; $w_\infty$ exists for all $r$, since
$\tfrac12w_r^2+\tfrac14w^4$ does not increase, and it has infinitely many zeros by Atkinson's oscillation
theorem \cite{Atk55} ($x=-\sigma=\ln r$, $w_{xx}+e^{2x}w^3=0$, $\int^\infty xe^{2x}dx=\infty$). These zeros are simple
and accumulate only at $\sigma=-\infty$, so for every $m$ the function $W(\cdot\,;A)$ has at least $m$
zeros once $A$ is large.

\emph{Every $n$ occurs.} Given $n\ge0$, let $A_n=\inf\{A\ge1:N(A)\ge n+1\}$. $N$ is not locally
constant at $A_n$, so $\beta(A_n)=0$. By the local formula, $N(A_n)\le n$ (from $A<A_n$, or $N(1)=0$
if $A_n=1$), and $n+1\le N(A)\le N(A_n)+1$ for some $A>A_n$ close to $A_n$. Hence $N(A_n)=n$, and
$W(\cdot\,;A_n)$ is regular at both ends with exactly $n$ zeros.

\emph{The zero set.} At fixed $t$ the map $x\mapsto n(t,x)$ is a diffeomorphism of $\mathbb R^3$ onto $S^3$
minus the point $(-1,0,0,0)$, at which $\mu=1$: its polar angle $\chi=\arctan(t+r)-\arctan(t-r)$
increases from $0$ to $\pi$ with $r$. Each level set $\mu=\mu_k\in(-1,1)$ is the torus
$|z_1|^2=(1+\mu_k)/2$ in $S^3$, invariant under the rotations about the $z$-axis, so its
preimage is a compact embedded torus of revolution; different $\mu_k$ give disjoint,
nested tori.
\end{proof}
Part (3) settles the question, left open by the variational argument,
whether every number of zeros occurs; uniqueness for a given $n$ is not
claimed. The following data are numerical and are not part of the
proposition.

\begin{remark}[numerical solutions with prescribed nodes]\label{rem:laddernum}
Shooting from $\mu=1$ finds, for $n=1,\dots,5$, a solution with exactly $n$ zeros,
with $W(1)=4.2994680$, $9.0148390$, $14.794937$, $21.467736$, $28.925904$; it is odd in
$\mu$ for odd $n$ and even for even $n$. The energies $\pi^2E(W)$ of the corresponding $u$
are $199.82$, $1419.15$, $5294.77$, $14261.5$, $31556.2$, against $\pi^2/2\approx4.93$ for $W\equiv1$. For $0<A<31$ the
function $\beta(A)$ of the proof vanishes exactly at $A=1$ and at these five values,
and $N(A)$ increases by one just to the right of each of them. \Nm
\end{remark}

At $t=0$ the nodal set $\mu=\cos2\vartheta_k$ of $W$ is a round torus of revolution about the
$z$-axis, $\bigl(\sqrt{x^2+y^2}-\csc\vartheta_k\bigr)^2+z^2=\cot^2\vartheta_k$; for $n=1$ it is
$(\sqrt{x^2+y^2}-\sqrt2)^2+z^2=1$. The solutions are even in $t$. At fixed $x$ they decay
like $t^{-2}$; near the light cone $|x|\approx t$, in particular on the $z$-axis and on the
expanding ring $z=0$, $\sqrt{x^2+y^2}=\sqrt{1+t^2}$, they decay only like $t^{-1}$. Their
energies and Morse indices are listed in the supplementary material.

\section{Comparison with Lie reductions}\label{sec:lie}
A subalgebra of symmetries with three-dimensional orbits reduces the
equation to an ordinary differential equation, one with two-dimensional
orbits to an equation in two variables. For each family we computed
numerically the algebra of conformal generators (for the cubic) or of
Poincar\'e generators and the dilation (otherwise) that leave the solution
invariant (supplementary material). \Nm\ For the two rows marked \Sm\ the
absence of symmetries is proved in exact rational arithmetic (supplementary
material, Section~S6).
\begin{center}
\small
\begin{tabular}{@{}p{0.43\textwidth}cp{0.33\textwidth}@{}}
\toprule
solution & dim & invariance algebra\\
\midrule
rotating wave \eqref{eq:rotwave}, constant $\alpha$ & 6 & $\mathfrak{so}(2,2)$: $P_0,P_z,J_{0z},D,K_0,K_z$\\
same on $\mathbb H^2$, straight spacelike line, constant phase & 6 & six-dimensional\\
cylinder $\rho^{-1/2}\Lambda$, source at rest & 4 & $P_0,P_z,J_{0z},D$\\
Liouville field \eqref{eq:liouville}, $f(Y)=Y^2$, source at rest & 3 & $P_0,J_{xy},D$\\
torus solutions, Example~\ref{ex:torus} & 3 & $\tfrac12(P_0+K_0),\ \tfrac12(P_3-K_3),\ J_{12}$\\
$u=\Omega\psi(\tau)$, Example~S13.5(ii) & 6 & $\mathfrak{so}(4)$\\
\midrule
rotating wave \eqref{eq:rotwave} & \textbf{0} & ---\\
cone field of an accelerated source, rotating frame & \textbf{0} & ---\\
cone field $4\sigma/(3\rho^2-\sigma^2)$, accelerated source, turning axis \Sm & \textbf{0} & ---\\
same on $\mathbb H^2$ with phase $\theta_0(\tau)$ & \textbf{0} & ---\\
cylinder \eqref{eq:cylfield} with phase $\varphi_0(t-\rho)$ & 1 & $P_z$ only\\
Liouville field \eqref{eq:liouville}, accelerated source \Sm & \textbf{0} & ---\\
\bottomrule
\end{tabular}
\end{center}
The pattern is uniform. Each family with arbitrary functions contains a
classical reduction as the special case of constant data (fixed axis,
constant phase, source at rest), and for generic data depending on the
characteristic phase the symmetry disappears altogether; symmetries are
restored only by special data, such as constant data or source trajectories
that are orbits of a one-parameter subgroup of the Poincar\'e group. Such solutions
can be described only in terms of differential constraints. The lifted
families of \S\ref{sec:lifting} are of the same kind. Conversely, the smooth solutions of
Example~\ref{ex:torus} are a classical reduction, by an abelian subalgebra of the
conformal algebra; what is new about them is the global analysis of
Proposition~\ref{prop:ladder}.

\section{Physical interpretation and stability}\label{sec:stability}
\paragraph{Null modulation is exact.} A solution of the plane branch
depends on $(t,x)$ only through $\varsigma=t+x$, so the configuration moves
along $-x$ at the speed of light while its parameters vary from one value
of $\varsigma$ to another; this is not a Lorentz image of a static solution.
For $u=u(y,z,\varsigma)$ and $F=-V'$ the energy density and momentum density are
\[
  T_{00}=u_\varsigma^2+\tfrac12\bigl(u_y^2+u_z^2\bigr)+V(u),\qquad |T_{0x}|=u_\varsigma^2,
\]
so $T_{00}-|T_{0x}|=\tfrac12(u_y^2+u_z^2)+V(u)$: the modulation enters only the
null part, transports purely lightlike energy-momentum and exerts no
back-reaction on the transverse profile. This is why it is exact at any
amplitude. The arbitrary functions are data prescribed on characteristic
surfaces (Goursat data) on the null hyperplanes $\varsigma=\mathrm{const}$,
not Cauchy data on $t=0$, so their freedom does not conflict with well-posedness
of the Cauchy problem.

\paragraph{Linear stability of unmodulated walls.} We consider only profiles with
constant parameters. Linearising \eqref{eq:main}
about a profile $u_0(y)$ homogeneous in the other variables and putting
$v=\chi(y)e^{-i\omega t+i\mathbf k\cdot\mathbf x}$ gives
$-\chi''-F'(u_0)\chi=E\chi$ with $E=\omega^2-\mathbf k^2$.

\begin{proposition}\label{prop:stability}
(i) For $F=u^3-u$ and $u_0=\sqrt2\,\mathrm{sech}\,y$ the operator is
$-\chi''+(1-6\,\mathrm{sech}^2y)\chi$, with levels $E=-3$, $E=0$ and continuum
$E\ge1$; the soliton is unstable for $\mathbf k^2<3$, with maximal growth rate
$\sqrt3$ at $\mathbf k=0$.
(ii) For $F=u-u^3$ and $u_0=\tanh(y/\sqrt2)$ the operator is
$-\chi''+(2-3\,\mathrm{sech}^2(y/\sqrt2))\chi$, with levels $E=0$, $E=3/2$ and
continuum $E\ge2$; the kink is linearly stable.
(iii) For the sine-Gordon equation, $F=-\sin u$, and the wall $u_0=4\arctan e^y$ of
Proposition~\ref{prop:smoothsg} the operator is $-\chi''+(1-2\,\mathrm{sech}^2y)\chi$, with the single
level $E=0$ and continuum $E\ge1$; the wall is linearly stable.
In all cases the translational mode $E=0$ has the null dispersion
$\omega^2=\mathbf k^2$.
\end{proposition}
\noindent The operators reduce to reflectionless P\"oschl--Teller operators
\cite{PT33} (supplementary material).

The arbitrary functions of the lift are the exact nonlinear continuation of the
gapless translational mode at null momentum (supplementary material,
Proposition~S13.6). All stability statements here are linear and
concern walls with constant parameters; they do not extend without separate
analysis to walls whose parameters are modulated by arbitrary functions, and
nonlinear (orbital) stability is not addressed.

\section{Status of the statements and open problems}\label{sec:status}
\paragraph{Status.} The markers are those of \S\ref{sec:conv}; \Qm\ refers to
Theorem~\ref{thm:umbclass}, Kerr's theorem, the Herglotz--Noether theorem, the
Liouville-type theorem of \cite{Bre84}, and the variational, rigidity and
oscillation theorems used in Proposition~\ref{prop:ladder}.

\paragraph{Open.}
\begin{enumerate}
\item The hyperbolic type of two constraints (two characteristic phases; the
  only case found is the radial reduction, and two light-cone phases from
  different curves solve the system of LDEs neither for the cubic nor for the exponential). Theorem~\ref{thm:liouvclass} for smooth,
  non-analytic twisting phases, and Liouville fields not of the form of
  Proposition~\ref{prop:liouvgen}.
\item The elliptic type of two constraints for complex phases that are not Kerr
  direction functions (Remark~S11.4), for powers and for the exponential.
\item Classification of the solutions of the weighted system \eqref{eq:plap}.
\item Linear stability of the lifted configurations with non-constant data
  and of crossing walls (Proposition~\ref{prop:smoothsg} with $N\ge2$).
\item Twisting \emph{spacelike} single constraints: is a rotating, non-deforming
  spacelike flow proportional to a conformal Killing vector? With it, every
  rank-one branch outside the classical ones would be a conformal symmetry.
\item Uniqueness, for each $n$, of the solution with $n$ zeros in
  Proposition~\ref{prop:ladder}(3) (numerically it is unique for $n\le5$).
\end{enumerate}

\section{Conclusion}
The geometry of the constraint fields almost fixes the nonlinearity, and
arbitrary functions enter through a characteristic phase. The solutions
that, for generic data, no symmetry reduction produces are: smooth bounded sine-Gordon walls,
jets and vortices whose directions and positions vary arbitrarily along a
light front, a single wall being linearly stable; the field of a point
source moving on an arbitrary curve, for the cubic, which radiates according
to an exact nonlinear Larmor formula; and the Liouville
field of such a source, which together with the lifted Liouville fields
exhausts the Liouville fields built from complex functionally invariant
phases (real analytic where their rays twist). Three constraints give, for the cubic, smooth
finite-energy solutions with any prescribed number of toroidal nodal surfaces.

\section*{Supplementary material}
The supplementary material contains the proofs and statements omitted here
(each omission is indicated where it occurs) and the computer-algebra and
numerical scripts for the statements marked \Sm\ and \Nm.

\section*{Acknowledgments}
The research was carried out within the state assignment of Ministry of
Science and Higher Education of the Russian Federation for Federal Research
Center for Information and Computational Technologies. During the preparation of this
work the author used the AI assistant Claude (Anthropic) for symbolic and
numerical verification of calculations, literature search and editing
of the text. All results were checked by the author, who takes full
responsibility for the content of the article.

\section*{Author declarations}
\paragraph{Conflict of interest.} The author has no conflicts to disclose.
\paragraph{Author contributions.} \textbf{O.~V. Kaptsov}: Conceptualization;
Formal analysis; Investigation; Methodology; Software; Validation;
Writing -- original draft; Writing -- review \& editing.

\section*{Data availability}
The data that support the findings of this study are available within
the article and its supplementary material.

\appendix
\section{Proof of Proposition~\ref{prop:radiation}}\label{app:radiation}
\begin{proof}[Proof]
(1) In the rest frame of the source at the moment of emission $\ell=(1,\mathbf n)$ with
$|\mathbf n|=1$ and $e=(0,\mathbf e)$, so $\mu=-\ell\cdot e=\mathbf n\cdot\mathbf e$ is the cosine of the angle
between the ray and the axis. Hence $3-\mu^2\ge2$ and $|W|\le W(1)=2$.
(2) From $\partial_\nu R^\kappa=\delta_\nu^\kappa-V^\kappa\ell_\nu$ and $V\cdot e=0$ we get
$\partial_\nu\sigma=-e_\nu-\rho\,(\ell\cdot\dot e)\,\ell_\nu$; with $\partial_\nu\rho=V_\nu+(A-1)\ell_\nu$ and
$A=\rho\,(\ell\cdot a)$ this gives $\nabla\mu=-\bigl[(\ell\cdot\dot e)+\mu(\ell\cdot a)\bigr]\ell-\rho^{-1}\bigl(e+\mu V-\mu\ell\bigr)$,
and $\nabla u=W'\nabla\mu/\rho-W\nabla\rho/\rho^2$ is the stated formula. Since $\dot e=-(e\cdot a)V+w$
and $\ell\cdot V=1$, $\ell\cdot\dot e=-(e\cdot a)+\ell\cdot w$, so $f$ is linear in $a$ and $w$.
(3) From $\ell\cdot e=-\mu$ and $\ell\cdot V=1$ one gets $\ell\cdot N=-W$, hence
$\partial u\cdot\partial u=-2fW\rho^{-3}+N\cdot N\rho^{-4}$ and
\[
  T^{\mu\nu}=\frac{f^2}{\rho^2}\ell^\mu\ell^\nu+\frac{f}{\rho^3}\bigl(\ell^\mu N^\nu+N^\mu\ell^\nu+W\eta^{\mu\nu}\bigr)+O(\rho^{-4}),
\]
the coefficients being constant along each ray. The middle term is annihilated
by $\ell_\nu$, so $T^{\mu\nu}\ell_\nu=O(\rho^{-4})$ and $T^{\mu\nu}V_\nu=f^2\ell^\mu/\rho^2+O(\rho^{-3})$. The normal
$\partial_\nu\rho=V_\nu+(A-1)\ell_\nu$ grows like $\rho$ through $A=\rho\,(\ell\cdot a)$, but the two estimates
still give $\rho^2T^{\mu\nu}\partial_\nu\rho=f^2\ell^\mu+O(\rho^{-1})$. In the coordinates $(\tau,\rho,\mathbf n)$ the
volume element is $\rho^2\,d\tau\,d\rho\,d\Omega$, so the four-momentum crossing the surface
$\rho=\mathrm{const}$ per unit $\tau$ tends to $\oint f^2\ell\,d\Omega$; its component along $V$ is $P$.
(Where $\ell\cdot a>0$ this surface is spacelike for $\rho>1/(2\,\ell\cdot a)$, since
$(\nabla\rho)^2=2A-1$; this does not affect the limit.) In the rest frame,
$\ell\cdot a=-\mathbf n\cdot\mathbf a$ and $\ell\cdot\dot e=\mathbf e\cdot\mathbf a-\mathbf n\cdot\mathbf w$, so
\[
  f=(\mathbf n\cdot\mathbf a)(\mu W)'-(\mathbf e\cdot\mathbf a)W'+(\mathbf n\cdot\mathbf w)W' .
\]
Write $\mathbf n=\mu\mathbf e+\mathbf n_\perp$ and $\mathbf a=a_\parallel\mathbf e+\mathbf a_\perp$; then
$f=a_\parallel\bigl[\mu(\mu W)'-W'\bigr]+\mathbf n_\perp\cdot\bigl(\mathbf a_\perp(\mu W)'+\mathbf wW'\bigr)$. Averaging over the
azimuth, $\int_0^{2\pi}(\mathbf n_\perp\cdot\mathbf b)^2d\varphi=\pi(1-\mu^2)|\mathbf b|^2$ and the mixed terms vanish; the
term in $\mathbf a_\perp\cdot\mathbf w$ has the coefficient $2\pi\int_{-1}^1(1-\mu^2)(\mu W)'W'\,d\mu=0$, the integrand being
odd. The remaining integrals over $\mu$ are elementary and give the constants.
\end{proof}

\end{document}